\documentclass[10pt,journal,twoside]{IEEEtran}
\usepackage[applemac]{inputenc}

\ifCLASSINFOpdf
\else
\fi
\usepackage[tight,footnotesize]{subfigure}
\usepackage{epsfig}
\usepackage{amsmath}
\usepackage{amsthm}
\usepackage{amsfonts}
\usepackage{amssymb}
\usepackage{graphicx}
\usepackage{graphics}
\usepackage{float}
\usepackage{url}
\usepackage{eurosym}
\usepackage{setspace}
\usepackage{subfigure}
\usepackage{balance}
\usepackage{dsfont}
\usepackage{amsmath}
\usepackage{bbm}

\def\ontop#1#2{\setbox0\hbox{#2}\copy0\llap{\raise\ht0\hbox{#1}}}

\theoremstyle{plain}
\newtheorem{thm}{Theorem}

\newtheorem{prp}{Proposition}

\theoremstyle{definition}
\newtheorem{defn}{Definition}

\theoremstyle{remark}

\makeatletter
\renewcommand*\env@matrix[1][\arraystretch]{%
  \edef\arraystretch{#1}%
  \hskip -\arraycolsep
  \let\@ifnextchar\new@ifnextchar
  \array{*\c@MaxMatrixCols c}}
\makeatother

\title{Track selection in Multifunction Radars:\\ Nash and correlated equilibria}

 \author{Nikola~Bogdanovi\'{c},  
        Hans~Driessen, 
        and~Alexander~Yarovoy
        \thanks{Nikola~Bogdanovi\'{c}, Hans~Driessen and Alexander~Yarovoy are with the Microwave Sensing,
        	Systems and Signals (MS3) group, Delft University of Technology, Mekelweg
        	4, 2628CD Delft, The Netherlands. E-mails: \{N.Bogdanovic, J.N.Driessen, A.Yarovoy\}@tudelft.nl. (Nikola Bogdanovi\'{c}, as corresponding author).}
\thanks{ This research has been partly supported by the SOS project funded by the European Commission (FP7 GA no. 286105).
}}

\usepackage{algorithm,algcompatible,amsmath}
\algnewcommand\INPUT{\item[\textbf{Input:}]}%
\algnewcommand\OUTPUT{\item[\textbf{Output:}]}%
\algnewcommand\phase{\item[\textbf{Pahse 1:}]}
\algnewcommand\INI{\item[\textbf{Initial state:}]}
\algnewcommand\prop{\item[\textbf{Proposed Graph Formation for NSPE Algorithm:}]}
\algnewcommand\phaseI{\item[\textit{\-\hspace{0.1cm} \underline{Phase I} - Adaptive coalitional graph formation:}]}
\algnewcommand\phaseII{\item[\textit{\-\hspace{0.1cm} \underline{Phase II} - NSPE diffusion strategy:}]}
\algnewcommand\period{\item[\textbf{Repeat periodically:}]}

\begin{document}
%
\maketitle
\begin{abstract}

We consider a track selection problem for multi-target tracking in a multifunction radar network from a game-theoretic perspective. The problem is formulated as a non-cooperative game. 
The radars are considered to be players in this game with utilities modeled using a proper tracking accuracy criterion
and their strategies are the observed targets whose number is known. Initially, for the problem of coordination, the Nash equilibria are characterized and, in order to find equilibria points, a distributed algorithm based on the best-response dynamics is proposed. Afterwards, the analysis is extended to the case of partial target observability and radar connectivity and 
heterogeneous interests  among radars. The solution concept of correlated equilibria is employed and a distributed algorithm based on the regret-matching is proposed.
The proposed algorithms are shown to perform well compared to the centralized approach of significantly higher complexity.

\end{abstract}
%


\section{Introduction }

Radar networks that employ multiple, distributed stations have attracted a lot of attention due to the improvements in tracking and detection performance they may offer over conventional, standalone radars. 
Furthermore, 
recent advances in sensor technologies enabled a large number of controllable degrees of freedom in modern radars. 
One such system is the 
Multifunction Radar (MFR), and it typically employs phased array antennas that allow
the radar beam to be controlled 
almost instantaneously~\cite{hero2011sensor}\nocite{sabatini_tarantino_1994}-\cite{richards_scheer_holm_2010}. Thus, the MFR is much more flexible than conventional, dedicated radars 
by being capable of performing multiple
functions simultaneously - volume surveillance, fire control, and multiple target tracking to name a few.
%
In this paper, we focus on the latter function~\cite{blackman1999design}\nocite{bar1990multitarget}\nocite{mallick2012integrated}\nocite{Mahler_book_2007}-\cite{Waveform_agile_4775880}; specifically, each MFR radar performs the track filtering of several targets.

The aforementioned flexibility introduces a need
to effectively manage available radar resources 
to achieve specified objectives while conforming to operational and technical constraints~\cite{Ding_4564804},~\cite{katsilieris2015sensor}. %
%
%
Even for a standalone 
MFR, the radar resource management plays a crucial role. 
Most of the existing approaches to MFR radar resource management 
roughly fit into the following two categories~\cite{charlish2011autonomous}\nocite{Fotios_AES_7376216}-\cite{narykov2013algorithm}. 
%
The first category consists of the rule-based techniques~\cite{van1993phased}\nocite{koch1999adaptive}-\cite{coetzee2005multifunction}, which control the resource allocation parameters indirectly, under low computational burden
. 
The main drawback of these techniques is that it is hard to say what performance 
can be achieved since it highly depends on the application scenario and on the sensors being deployed. 
%
The other category is related to the methods that formulate the problem as an optimization one; and thus, they may achieve the optimal performance, 
see~\cite{hero2011sensor},~\cite{Waveform_agile_4775880},~\cite{Hans_1591902}\nocite{hansen2006resource}-\cite{Djonin_4915772} and the references therein.
In the network setting, which is  the focus of this paper, the first category of approaches is difficult to be extended, while the second one may involve excessive complexity due to the network dimension~\cite{Covariance_control_Kal_1145739}-\cite{Griffiths_7527905}.
Thus, to reduce such complexity, one may aim to find, in either 
centralized or distributed way, a close-optimal solution to the radar management problem that is considered, see e.g.,~\cite{Griffiths_7527905}\nocite{Severson_AEES_6965777}-\cite{Zhi_Tian_AEES_6324736}.
%
%
In this work, we propose a distributed approach based on game theory 
so as to 
model track selection for multi-target 
track filtering in an MFR network.

Game theory is the mathematical study of conflict and cooperation between intelligent rational decision-makers~\cite{shoham2008multiagent}. 
Apart from economics and political sciences,
over the last decade  game theory (GT) is being applied to 
control, signal processing and wireless communications, mainly due to the issues dealing with networking~\cite{Marden_4814554}\nocite{Han_2012_book}\nocite{bacci_2015game}\nocite{Saraydar_983324}\nocite{Larsson_4604732}\nocite{scutari2008competitive}\nocite{ch1_j_Feleg4907463}
-\cite{GT_icassp2015}. 
More recently, GT has been applied to solve certain radar problems,  mostly related to the 
multiple input multiple output (MIMO) radar networks. 
%
For instance, the problem of waveform design has been investigated~\cite{g_2012_game}\nocite{Nehorai2016jointly}-\cite{code_design_piezzo2013non};  in~\cite{g_2012_game} 
by formulating a two player zero-sum (TPZS) game 
between the radar design engineer and an opponent, 
and in~\cite{code_design_piezzo2013non} by a potential game 
in which the radars choose among the pre-fixed transmit 
codes. 
%
%
%
Next, the interaction between a smart target and a MIMO radar has also been modeled as a TPZS 
game~\cite{Jammer_6025317}. 
%
Furthermore, the problem of transmission power 
control 
was addressed 
by using non-cooperative GT and Nash equilibria in~\cite{Bacci_6250454} 
 and  by employing a coalitional game theoretic solution concept called the Shapley value in~\cite{Chen_coop_game_7104065}. 
Although not dealing with radar management, a useful work in~\cite{chavali2013concurrent} related to the multi-target tracking application uses correlated equilibrium to solve the data association problem at a single radar. 
Finally, the works in~\cite{charlish2011autonomous},~\cite{charlish2012multi}-\cite{Charlish_AES_7272863} utilize a market mechanism, called the continuous double auction, in order to choose the global optimum parameters for each individual task given the global (finite) resource constraint. 
The method provided a superior performance over its competing
heuristic-based algorithms; however, its main drawback is in the
implementation complexity~\cite{Griffiths_7527905}.

In this paper we apply game theory to multi-target track filtering in an MFR network and extend the initial results from~\cite{track_selection_7472254}. %
The main contributions of this article are 
the following ones: 
\begin{itemize}
	\item A new formulation of the track selection problem for a multi-target tracking scenario in a resource-limited MFR network using the non-cooperative games is proposed. 
	\item The 
	track selection %
	problem is analyzed using the Nash equilibria of the underlying coordination game for the setting with full target observability and radar connectivity as well as the homogeneous interests (target priorities) of radars. Also, to solve the  
	problem in a distributed manner, a low-complexity algorithm based on the best-response dynamics is proposed.
	\item The 
	track selection %
	problem is extended to the case of partial target observability and radar connectivity and 
	heterogeneous interests  among radars. Due to the particularities of this case, the solution concept of correlated equilibria is employed and a distributed algorithm based on the regret-matching is proposed.
\end{itemize}

The structure of the paper is as follows. Section~\ref{Sec_2_GT_overview} provides some background on game theory and the solution concepts employed in this paper. The problem formulation is given in Sec. \ref{Sec_3_Problem_form}. 
%
 Next, Sections~\ref{coordination}
 -\ref{coord_CONflict} deal with the analyses of the scenarios where the observability and connectivity conditions as well as radar interests are being equal and heterogeneous, respectively. 
%
%
Specifically, in Sec. \ref{coordination} the former scenario 
 is modeled 
 as a coordination game, its Nash equilibria are characterized in terms of their existence conditions and efficiency, and a distributed algorithm based on the best-response dynamics is proposed. 
 On the other hand, Section~\ref{coord_CONflict} 
 provides a distributed algorithm tracking the set of correlated equilibria points.  
 In Section~\ref{simulat}, the 
 effectiveness of the proposed algorithms is 
 demonstrated via computer simulations. 
 Finally, Section~\ref{sec:conc} summarizes the work.

\begin{figure}
	\centering
	\includegraphics[width=0.725\linewidth]{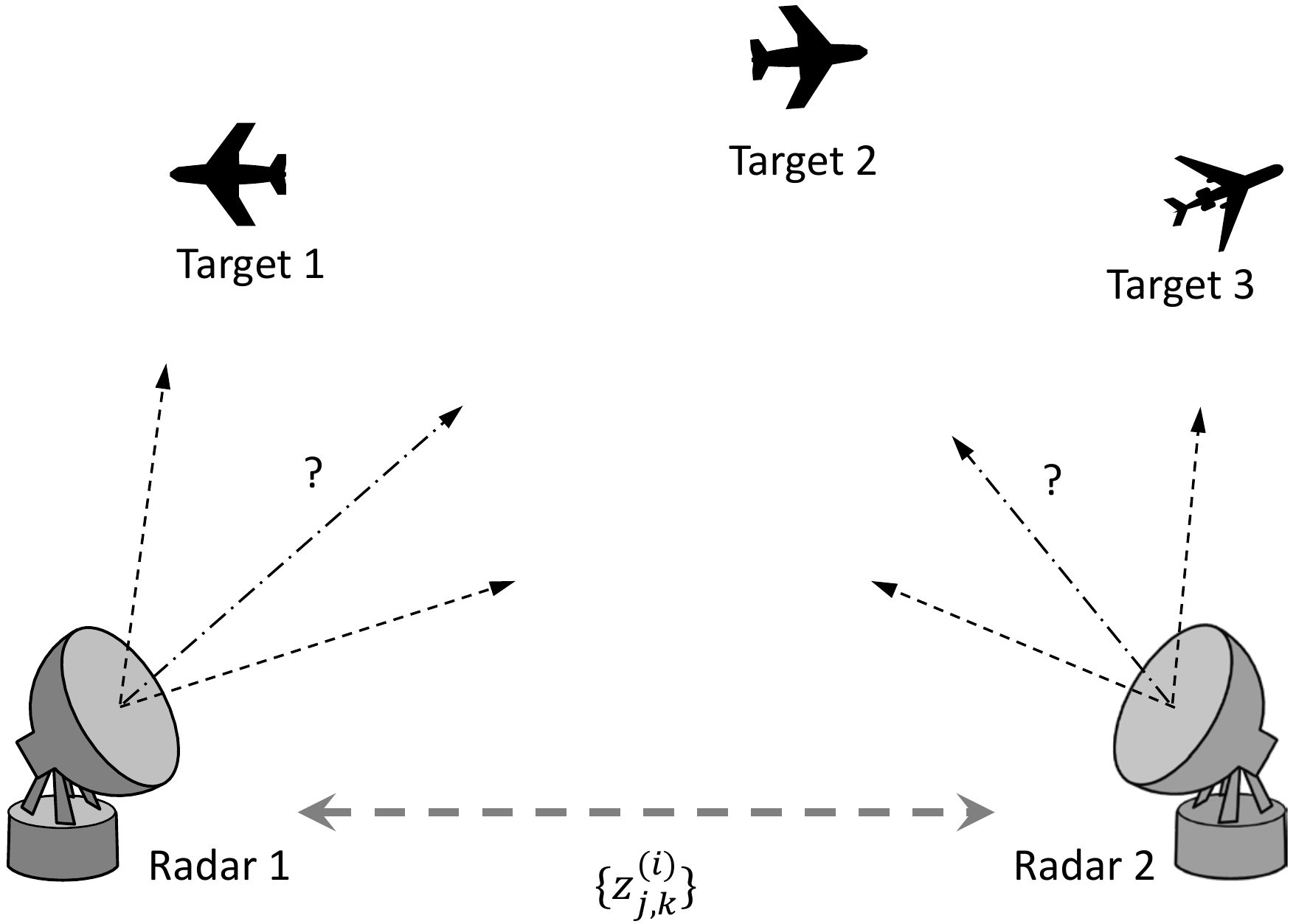}
	\caption{A track selection problem in multi-target tracking.}
	\label{fig:radar_network_1_1}
\end{figure}

\section{Brief preliminaries on game theory}
\label{Sec_2_GT_overview}

In this section, we provide notation and recall some formal definitions and solution concepts
related to game theory that will be used throughout the paper.
The focus is put on noncooperative game theory, the dominant branch of game theory, and specifically
on so-called normal-form games~\cite{shoham2008multiagent}.

\begin{defn}
A finite, $N$-person normal-form game is a tuple $\Gamma=(\mathcal{N}, \mathcal{S}, u)$, where:
\begin{itemize}
	\item $\mathcal{N}$ is a finite set of $N$ players.
	\item $\mathcal{S} = \mathcal{S}_1 \times \cdots  \times \mathcal{S}_N$, where $\mathcal{S}_i$ is a finite set of actions (strategies) available to player $i$, $\forall i \in \mathcal{N}$. Each vector $s = (s_1, \ldots , s_N) \in \mathcal{S}$ is called an action (strategy) profile.
	\item $u = (u_1, \ldots , u_N)$ where $u_i \colon \mathcal{S} \rightarrow \mathbb{R}$ is a real-valued utility (or payoff) function for  player  $i$, $\forall i \in \mathcal{N}$ .
\end{itemize}
\end{defn}

To reason about multiplayer games, one can rely on solution concepts, i.e., principles according to which interesting outcomes of a game can be identified. 
Some fundamental concepts, which will be used throughout this paper, are described in the sequel.
A basic and the most widely accepted one is the celebrated Nash equilibrium. %
Formally, in case where players make deterministic choices (pure strategies)
the Nash equilibrium is defined as follows%
~\cite{shoham2008multiagent}.
 
\begin{defn}
 A strategy profile $s = (s_1, \ldots , s_N)$ is a pure-strategy \textit{Nash
equilibrium} if, for all players $i$  and for all strategies $s^{\prime}_i \neq s_i$, it holds that 
$$u_i (s_i, s_{-i}) \geq  u_i (s^{\prime}_i, s_{-i}) ,$$
where $s_{-i} = (s_1, \ldots , s_{i-1}, s_{i+1}, \ldots , s_N)$ is defined as a strategy
profile $s$ without player $i$'s strategy. 
\end{defn}
Otherwise stated, a Nash
equilibrium (NE) is a state of a non-cooperative game where no player can unilaterally improve its utility by taking a different strategy, if the other players remain constant in their strategies.  

Next, we define the concepts of Pareto domination and Pareto optimality.

\begin{defn} 
Strategy profile $s$  \textit{Pareto dominates} strategy profile $s^{\prime}$ if  $\forall i \in \mathcal{N}$, $u_i(s) \geq u_i(s^{\prime})$, and there exists some $j \in \mathcal{N}$ for which $u_j(s) >  u_j(s^{\prime})$. Also, strategy profile $s$ is  \textit{Pareto optimal} 
if there does not exist another strategy profile $s^{\prime} \in \mathcal{S}$ that Pareto dominates $s$.
\end{defn}

To evaluate the (in)efficiency  of NE there is a notion called the price of anarchy, which is defined as the ratio of a centralized solution to the worst-case equilibrium in terms of 
the utility sum that is in economics literature 
known as "social welfare".
\begin{defn} 
	The price of anarchy (PoA) is given as
	 \begin{equation*}\label{eps:CE_def_1} 
{\rm PoA} = {\max\limits_{s\in \mathcal{S}} \, \,\sum_{i \in \mathcal{N}} u_i(s)\over \min\limits_{s\in \mathcal{S}^{\rm NE}} \, \,\sum_{i \in \mathcal{N}} u_i(s)} \, ,
\end{equation*}
where $\mathcal{S}^{\rm NE}$ is the set of Nash equilibria of the game.
\end{defn}
Note that in case where the equilibria are fully efficient, the PoA is equal to 1.

Finally, we define 
the notion of correlated equilibrium, which is a generalization of Nash equilibrium~\cite{Aumann197467}, \cite{hart2013simple_book}. 
\begin{defn}   A \textit{correlated equilibrium} consists of a probability vector\footnote{A probability vector is a vector whose coordinates are all nonnegative and sum up to 1.} $\pi$ on $\mathcal{S}$ such that the following is satisfied,  $\forall i \in \mathcal{N}$ and $\forall s_{i}, s^{\prime}_i \in \mathcal{S}_i$:
 \begin{equation*}\label{eps:CE_def_1} 
 \sum\limits_{s_{-i}\in \mathcal{S}_{-i}} \pi(s_{i}, s_{-i})[u_{i}(s_{i}, s_{-i})-u_{i}(s_{i}^{\prime}, s_{-i})]\geq 0. 
 \end{equation*}
\end{defn}

To interpret the inequality above, let us first divide it by the marginal probability $ \pi(s_{i})$, which yields:
 \begin{equation*}\label{eps:CE_def_2}
  \sum\limits_{s_{-i}\in \mathcal{S}_{-i}} \pi(s_{-i} \vert s_{i})[u_{i}(s_{i}, s_{-i})-u_{i}(s_{i}^{\prime}, s_{-i})]\geq 0.
 \end{equation*}
Thus, an intuitive interpretation of correlated equilibrium is as follows: Suppose that a strategy profile $s \in \mathcal{S}$ is chosen at random, e.g., by some virtual referee, according to the joint distribution $\pi$.  Each player $i$ is then given, by the "referee", its own recommendation $s_i$. The inequality above means that player $i$ cannot obtain a higher expected utility by selecting strategy $s_{i}^{\prime}$ instead of the "recommended" one, i.e., $s_{i}$. %
Also, in every finite game, the set of correlated equilibria is nonempty, closed and convex. 


\section{Problem formulation section}%
\label{Sec_3_Problem_form}

Let us consider a network of MFR radars which aims at tracking several targets, e.g., see Fig.~\ref{fig:radar_network_1_1}. Let $\mathcal{N}$ denote the set of $N$ radars and $\mathcal{T}$  denote the set of $T$ targets. We consider that the position of each radar node\footnote{In this paper, we use the terms radar and node interchangeably.} $i \in \mathcal{N}$ is 
 known. 
Although there are works in the tracking literature that consider unknown number of targets, e.g.,~\cite{mallick2012integrated}, \cite{Mahler_book_2007}, in this work we focus on the case where the number of targets at each time instant is 
known. 
%
%
%
The current positions of targets are assumed to be known 
approximately. 
Also, the targets are supposed to be well-separated; thus the data association problem is trivial and different transmission beams are required so as to illuminate distinct targets.
Furthermore, assume that there is no 
central processing node to perform track filtering; 
in other words, 
fusion is done at each radar node.

Next, the 
dynamics of each target $j \in \mathcal{T}$, at each discrete time $k$, are represented using the so-called white noise constant velocity model~\cite{blackman1999design},~\cite{Hans_1591902} given by 
\begin{align} 
	\mathrm{x}_{j,k}  &= F\cdot \mathrm{x}_{j, k-1} + w_{j, k-1} \\ 
	z_{j,k}^{(i)} &=  h^{(i)}_j(\mathrm{x}_{j,k}) + \nu^{(i)}_{j,k}
\end{align}
where
\begin{itemize}
	\item[-] the state vector $ \mathrm{x}$ for each target $j$ is comprised of the two dimensional coordinates and velocity\footnote{Although here we assume a two-dimensional case, the extension to a three-dimensional case is straightforward.}, i.e., $\mathrm{x}_j=[ x_j ,y_j , v_{j,x}, v_{j,y}]^T$ where $[\cdot]^T$ stands for the transposition of the argument,
	\item[-] $F$ is a $4\times4$ matrix corresponding to the deterministic target dynamics given as 
	\begin{align}
	F=\left[\begin{smallmatrix} 
	1 & t_u \\
	0 & 1
	\end{smallmatrix}\right] \otimes I_2
	\end{align} 
	with $\otimes$ being the Kronecker product, $I_2$ stands for a $2\times2$ identity matrix and $t_u$ is the update time that is fixed,
	\item[-] the process noise $w$ is Gaussian with zero mean and covariance 
	\begin{align}
	Q=\sigma^2_w \cdot \left[\begin{smallmatrix} 
	t_u^3/3 & t_u^2/2 \\
	t_u^2/2 & t_u
	\end{smallmatrix}\right] \otimes I_2 
	 \end{align}
	where $\sigma^2_w$ models maneuverability,
	\item[-] the measurement vector $z_{j,k}^{(i)}$, at each radar $i \in \mathcal{N}$, consists of range and azimuth, i.e., $z_{j,k}^{(i)}=\left[r_{j,k}^{(i)}  , a_{j,k}^{(i)}\right]^T$,
	\item[-] the nonlinear transformation  $h^{(i)}_j(\mathrm{x}_{j})$ is given by 
	\begin{align}
	h^{(i)}_j(\mathrm{x}_{j})=\begin{bmatrix} 
	\sqrt{(x_j-x_i)^2 + (y_j-y_i)^2} \\
	\mathrm{arctan}((y_j-y_i)/(x_j-x_i)) \end{bmatrix} , 
	\end{align}
	\item[-] the measurement noise $\nu^{(i)}_{j}$ is zero-mean Gaussian with covariance $R_{j,i}=\mathrm{diag}\left\{[\sigma^{(i)}_{r_j}]^2,[\sigma^{(i)}_{a_j}]^2\right\}$.
\end{itemize}

The radars have limited time budget in the sense that they cannot take measurements of %
all targets 
during the same time slot. 
The number of measurements per scan that each radar can make is given by $m < \vert\mathcal{T}\vert$. 
Since there is no central entity that may coordinate actions of the radars, 
a distributed solution is needed. %

Furthermore, radars may experience different target observability conditions; thus, the set of the targets that are observable at each radar $i$ is denoted by $\mathcal{T}_i$, and it satisfies $\mathcal{T}_i \subseteq \mathcal{T}$.
The interaction among the radars is existing but limited to sharing 
the measurements $\{z_{j,k}^{(i)}\}$ 
related to the  previously selected targets. 
The communication neighborhood of any particular radar $i$, together with radar $i$, is denoted as $\mathcal{N}_i$, where $\mathcal{N}_i \subseteq \mathcal{N}$.
The number of transmissions each radar $i$ collects from its neighborhood $\mathcal{N}_i$, and which are related to some target $j \in \bigcup_{i \in \mathcal{N}_i}\mathcal{T}_i $, 
is denoted as 
$m^t_j(i)$. %
For notational simplicity, in the rest of this section we drop the index $j$ for targets where no confusion is possible.

At each radar $i$ and for each target $j$, the tracking process is performed by 
an Extended-Kalman Filter (EKF). Firstly, the prediction step occurs, i.e.,
\begin{align}
\mathrm{x}_{k\vert k-1}  &= F\cdot \mathrm{x}_{k-1\vert k-1}  \label{eq_predict_1}\\ 
P_{k\vert k-1} &=  F P_{k-1\vert k-1} F^T + Q  \label{eq_predict_2}
\end{align}
where $\mathrm{x}_{k\vert k-1}$ and $P_{ k\vert k-1}$ are the state estimate and the error covariance matrix for time step $k$ given all measurements till time step $k-1$. 
Then, the updating step takes place where each available measurement for target $j$ of some radar $n\in \mathcal{N}$ 
is used in a cyclic manner. In particular, for  each $p\in \{1,\ldots, m_j^t(n)\}$, 
\begin{align}
K_{k}^{(p)}  &= P_{k\vert k}^{(p-1)} [H^{(p)}_{k,n}]^T \left(H^{(p)}_{k,n} P_{k\vert k}^{(p-1)} [H^{(p)}_{k,n}]^T  +R_{n}\right)^{-1}   \label{eq_update_1}\\ 
\mathrm{x}_{k\vert k}^{(p)}  &= \mathrm{x}_{k\vert k}^{(p-1)} +K_{k}^{(p)} \left( z_{k}^{(n)} - h^{(n)}\left(\mathrm{x}_{k\vert k}^{(p-1)}\right)\right) \label{eq_update_2}\\ 
P_{k\vert k}^{(p)}  &=  \left(I- K_{k}^{(p)} H^{(p)}_{k,n}\right) P_{k\vert k}^{(p-1)} \label{eq_update_3}
\end{align}
where $P_{k\vert k}^{(p)}$ denotes the error covariance matrix after $p$ incremental updates at the same time step $k$, with 
$P_{k\vert k}^{(0)}=P_{k\vert k-1}$ and $\mathrm{x}_{k\vert k}^{(0)}=\mathrm{x}_{k\vert k-1}$.
The linearized measurement matrix of radar $n$ at time $k$ is $H^{(p)}_{k,n}= \partial h^{(n)}/ \partial \mathrm{x}$ evaluated at $\mathrm{x}_{k\vert k}^{(p-1)}$.
Note that, due to the fact that the coordinates $(x_n, y_n)$ of each radar $n \in \mathcal{N}$ are  
known, the radars do not need to exchange $\{H_{k,n}\}$ matrices in order to implement the algorithm above.

	In the following, we study a natural game theoretic variant of this problem. Specifically, we assume that the radars are 
autonomous decision-makers interested in optimizing their own tracking performance. In other words, the selections of each radar are autonomous in the sense that there is no entity to tell radars what to do in a hierarchical type of structure, nor is there any negotiation among radars. We analyze two indicative scenarios with respect to the observability conditions, communication topology as well as the radars' interests:
\begin{itemize}
\item[i)] \textit{Scenario 1}:  a scenario where 
each radar $i$ observes all targets, i.e., $\bigcap_{i \in \mathcal{N}}\mathcal{T}_i=\mathcal{T}$, 
 communicates with all neighbors (all radars communicate through the full graph), i.e., $\mathcal{N}_i = \mathcal{N}$, and is interested in tracking all targets in $\mathcal{T}$ (all targets have the same importance). 
\item[ii)] \textit{Scenario 2}: a more general scenario where the radars do not necessarily have the same target interests and where the observability and communication equalities above (full observability and full connectivity conditions) do not need to hold, i.e., $\exists i \in \mathcal{N} \mid \mathcal{N}_i \subset\mathcal{N}  \lor \mathcal{T}_i \subset\mathcal{T}  $.
\end{itemize}
%
%
%
For both scenarios, the fact that each radar (or the radar operator) autonomously and rationally decides to track the targets that increase its utility can be modeled as a 
one-stage non-cooperative game in normal form, which is the most fundamental representation type in game theory~\cite{shoham2008multiagent}. In following two sections, we analyze the track selection problem in each scenario separately. 


\section{Scenario 1: the problem of coordination}
\label{coordination}

Firstly, note that there are many classes of normal-form games; however, due to the particularities of the scenario considered, in this section we focus on coordination games, which do not rest solely upon conflict among players. Instead, as their name suggests, more emphasis is put on the coordination issue where players may have an incentive to conform with or to differ from what others do. In the latter case, this kind of games are usually called \textit{anti-coordination games}~\cite{shoham2008multiagent}, \cite{nisan2007algorithmic}-\cite{bramoulle2007anti}.

\subsection{Game-theoretic model}
\label{subsec_4_1_GT_model}

We assume that the players are rational and their objective is to maximize their payoff, i.e., the tracking accuracy of all targets. %
Formally, the 
track selection game $\Gamma^{(1)}=(\mathcal{N}, \mathcal{S}, u)$ has 
the 
subsequent components:
\begin{itemize}
	\item The players are the radars represented by the set $\mathcal{N}$.
	\item The strategy of each radar $i$ is represented by a $T$-tuple $s_i =(s_{i,1}, s_{i,2}, \ldots , s_{i,T})$ where $s_{i,j} =a$ 
	if radar $i$ devotes $a$ transmission beams to a target $j$, with $a\leq m$. Each strategy-tuple has at most $m$ transmissions, i.e., $\sum_{j=1}^{T} s_{i,j} \leq m$. 
	Also, note that $$m^t_j(i)=m^t_j=\sum_{i=1}^{N} s_{i,j}.$$ 
	\item The utility for each radar $i$ is given by
	\begin{gather}\label{eps:GT_model_eq1}
		\begin{split}
			u_i (s_i, s_{-i}) =     \sum_{j=1}^{T}  \mathrm{gain}_j(m^t_j) \, , 
		\end{split}
	\end{gather}
	where the term 
	$\mathrm{gain}_j(m^t_j)$ represents the  tracking accuracy gain for target $j \in \mathcal{T}$ and it is defined by
	\begin{equation}\label{eps:gain_eq}
		\mathrm{gain}_j (m^t_j) =
			\mathrm{Tr} \left\{ P_{j, k\vert k-1}  - P_{j, k\vert k}^{(m^t_j)} \right\} 
	\end{equation}
	where  $\mathrm{Tr}\{\cdot\} $ stands for the trace operator and  all radars are assumed to have the same initial guesses $\mathrm{x}_{j, 0\vert 0}$ and $P_{j, 0\vert 0}$.
\end{itemize}

In other words, the strategy of radar $i$ defines the number of transmissions per each target, at a given time slot, see Fig.~\ref{fig:example_matrix}.
Due to the fact that radars share their measurements
, their tracking accuracy gains for a specific target are 
dependent on all radars' measurements related to that target. 

\begin{figure}[!t]
	\centering
	\includegraphics[width=0.21\textwidth]{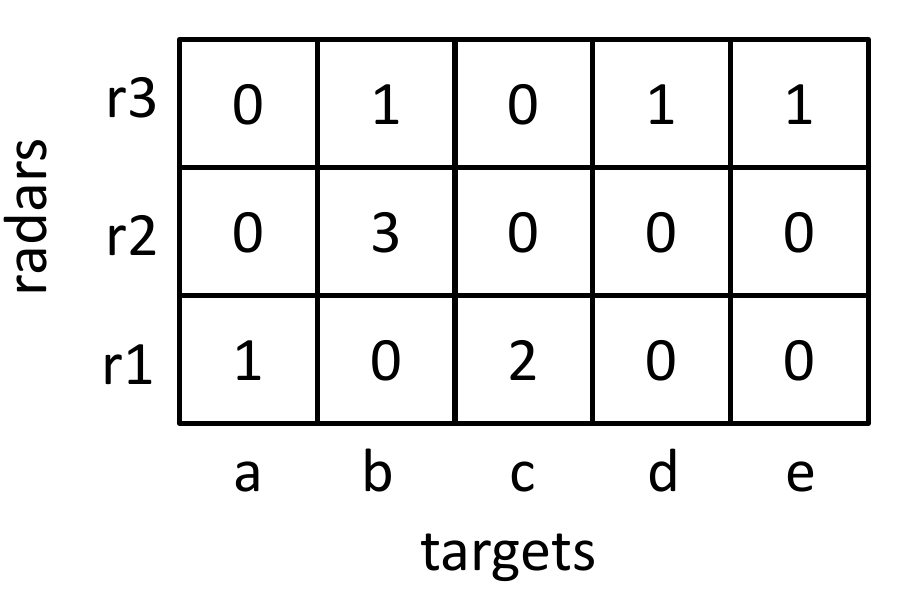}
		\caption{
			An 
			example strategy profile displayed as a matrix, for  
			$\mathcal{T}=\{a, \ldots, e\}$, and $\vert\mathcal{N}\vert=m=3$. }
	\label{fig:example_matrix}
\end{figure}

\begin{figure}[!t]
	\centering 
	\subfigure[]{
		\includegraphics[width=0.24\textwidth]{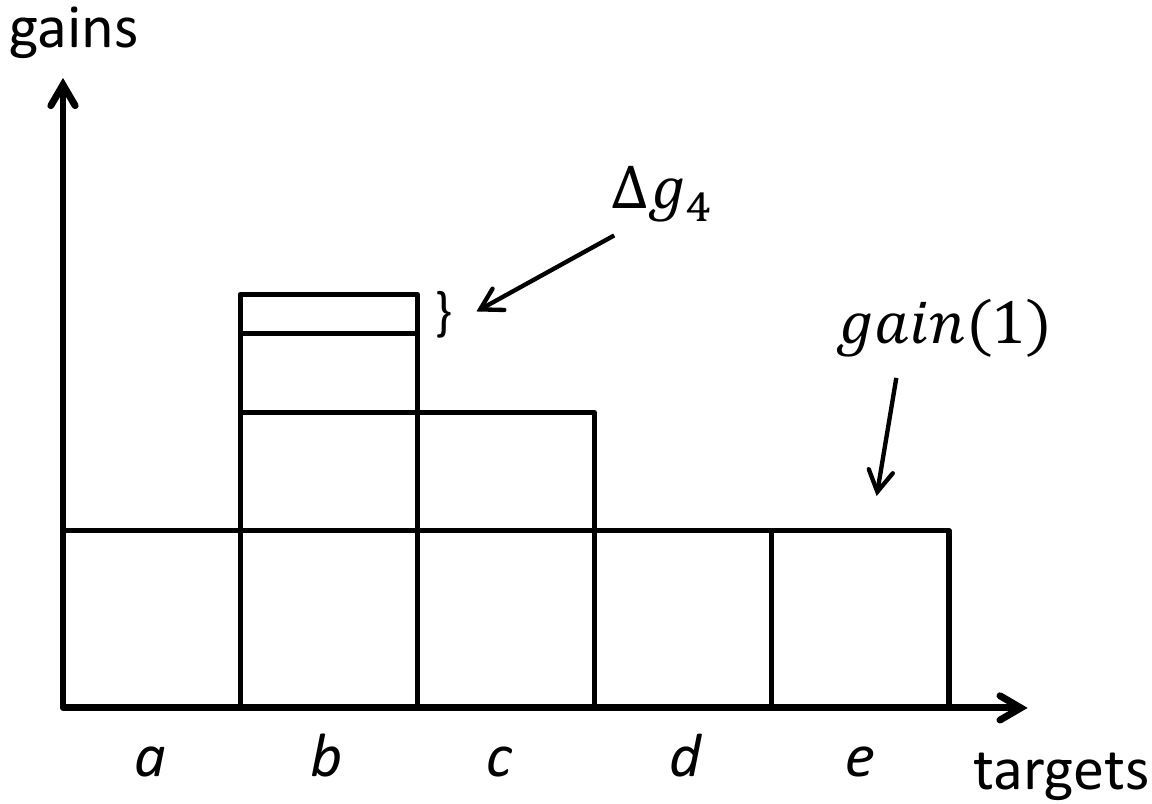}
		\label{fig:example1a}
	} \hspace{0.1cm}
	\subfigure[]{
		\includegraphics[width=0.24\textwidth]{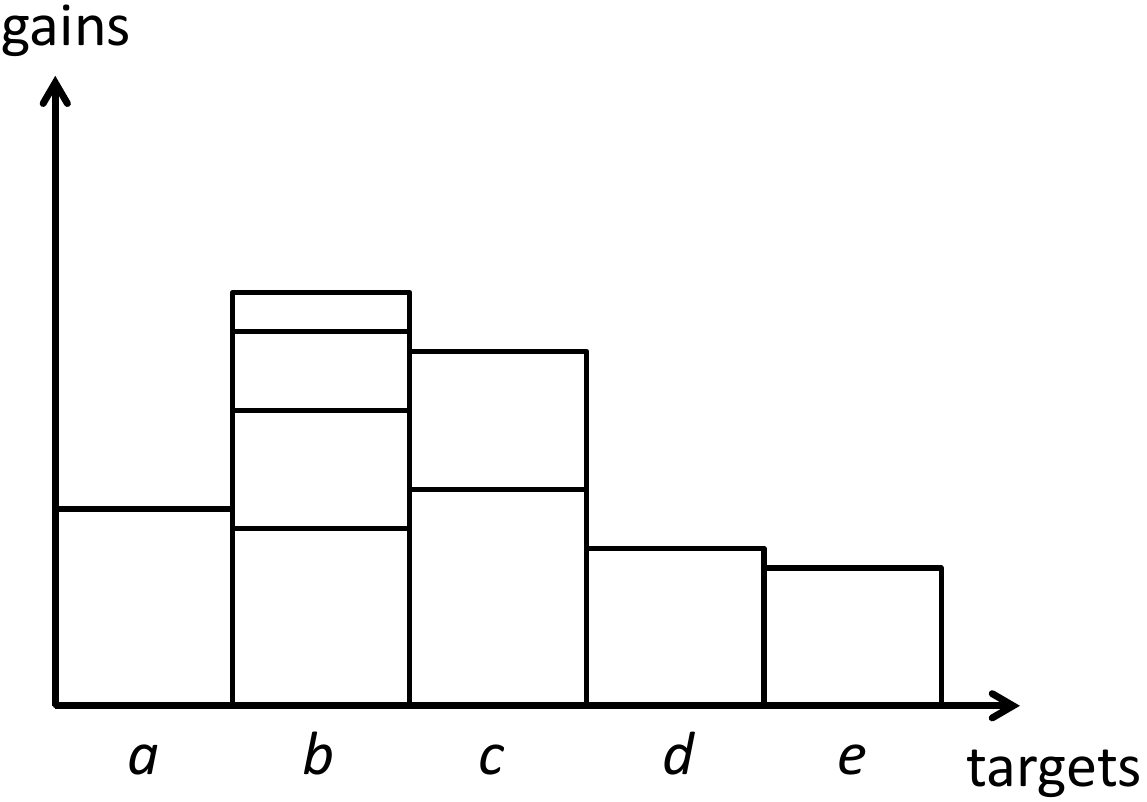}
		\label{fig:example1b}
	}
	\caption{An 
		example of a 
		track allocation in terms of gains per target where 
		$\mathcal{T}=\{a, \ldots, e\}$, and $\vert\mathcal{N}\vert=m=3$. 
		Each box represents a gain increment due to a measurement, and the number of measurements per target, $m^t_j$, varies between 1 and 4 across targets in $\mathcal{T}$
		. In case~\subref{fig:example1a}, the gains are equal for the same number of measurements, while in case~\subref{fig:example1b} they differ. 
	}  
	\label{fig:example1_1}
\end{figure}

Note that the gain in~\eqref{eps:gain_eq} can be expressed as 
\begin{align}
 \mathrm{gain}_j\left(m^t_j\right)= \sum_{p=1}^{m^t_j}  \Delta g^{(j)}_{p},
 \end{align}
where 
\begin{align}
\Delta g^{(j)}_{p}=\mathrm{Tr}\{ P_{j, k\vert k}^{(p-1)}  - P_{j, k\vert k}^{(p)} \}
\end{align}
and $\Delta g^{(j)}_{1}=\mathrm{gain}_j(1)$.
To analyze the proposed game, we proceed by adopting the following practical assumptions, for all $j \in \mathcal{T}$ and $p \in \{1, \ldots, m^t_j\}$:
\begin{itemize}
	\item \textit{Assumption 1}: the gain function in~\eqref{eps:gain_eq} is increasing in the number of measurements, i.e., $\Delta g^{(j)}_{p} > 0$, 
	\item \textit{Assumption 2}: estimation accuracy gain increment $\Delta g^{(j)}_{p}$ decreases as the order of measurements $p$ grows, i.e.,  $\Delta g^{(j)}_{p} > \Delta g^{(j)}_{p+1}$.
\end{itemize}

Finally, 
the following two cases are analyzed:
%
\begin{itemize}
	\item[a)] $\Delta g^{(j)}_{p}=\Delta g_{p}$,  
	for all $j \in \mathcal{T}$ and $p \in \{1, \ldots, m^t_j\}$
	\item[b)] $\Delta g^{(j)}_{p} \neq \Delta g^{(\ell)}_{p}$, 
	for $j \neq \ell$, 
	and $\mathrm{min}_{j \in \mathcal{T}}\Delta g^{(j)}_{p} >  \mathrm{max}_{j \in \mathcal{T}} \Delta g^{(j)}_{p+1}$.
\end{itemize}
Case a) represents an idealistic case where all nodes would have very similar measurements among themselves and related to all targets, see Fig~\ref{fig:example1_1}\subref{fig:example1a}. A more realistic scenario, corresponding to case b), is illustrated in Fig~\ref{fig:example1_1}\subref{fig:example1b}. 
In the next subsection, we characterize the Nash equilibria of the aforementioned cases.


\subsection{Nash equilibria}
\label{Sec_Nash}

Generally, in a coordination game, there are multiple NE. 
If the players have the same payoffs
, and the equilibria are equal, the game 
is a \textit{pure} coordination one.
In fact, in such a game, all NE are Pareto optimal. On the other hand, in a  \textit{ranked} one, 
the NE differ and usually there is only one Pareto optimal equilibrium~\cite{book_rasmusen_2001}.

Now, the main findings related to the NE for cases (a) and (b) are provided.

\begin{prp}
	\label{b_slucaj}
	The game for case (a) has PoA $=1$, and  any track assignment is a Nash equilibrium, if 
$\sum_{j=1}^{T} s_{i,j} = m$, and if
	\begin{itemize}
		\item  $m^t_j \leq 1$, $\forall j \in \mathcal{T}$, for a scenario where $N \cdot m \leq T$
		\item  $\mathrm{max}_{j, \ell \in \mathcal{T}} \{ \vert m^t_j - m^t_{\ell}\vert \} \leq 1$, $\forall j,\ell \in \mathcal{T}$, for a scenario where $N \cdot m > T$.
	\end{itemize} 
\end{prp}

\begin{proof}
Firstly, 
let us assume 
that there is a radar $i$ such that $\sum_{j=1}^{T} s_{i,j} < m$ and that the corresponding $s^{\ast}$ is an NE. 
Then, radar $i$ can change 
its strategy by taking an additional measurement. 
Due to the fact that  
the radar's gain function in~\eqref{eps:gain_eq} is increasing in the number of measurements, its utility will be increased. 
But that contradicts our initial assumption that  $s^{\ast}$ is an NE; thus, 
%
%
%
as per our intuition, each radar should make all possible transmissions toward the target(s) at each time instant. 
%
%
Next, 
note that if the total number of measurements is less than or equal to the number of targets, 
the radars are 
worse off if more than one measurement in total is devoted to the same target. 
Also, due to the structure of gain function, NE are precisely $\frac{T!}{\left(T -N \cdot m\right)!}$ outcomes in which each measurement is devoted to a distinct target.
On the other hand, if $N \cdot m > T$, the corresponding condition 
states that all targets should be covered as equally as possible. Here, each NE corresponds to a balanced allocation. For instance, the allocation in Fig~\ref{fig:example1_1}\subref{fig:example1a} is not an NE since the payoffs can be increased if some player moves its measurement from target $b$ to any other target. 
Finally, %
since the gain of any target is the same for the same number of measurements,  
%
the game appears to be 
a pure anti-coordination one. 
Thus, every NE is also Pareto optimal, which finally implies that PoA=1.
\end{proof}

\begin{prp}
	\label{less_case_c}
The game for case (b) has PoA $>1$, and any track assignment is an NE, if 
$\sum_{j=1}^{T} s_{i,j} = m$, and if,
	\begin{itemize}
		\item  
		for a scenario with $N \cdot m \leq T$,
		each radar chooses its 
		most accurate target that has not been selected,
		\item for 
		$N \cdot m > T$,
		the first $\left \lceil{\frac{N \cdot m}{T}}\right \rceil-1$ levels are filled in, i.e., $m_{j}^t \geq \left \lceil{\frac{N \cdot m}{T}}\right \rceil-1$,  $\forall j \in \mathcal{T}$,  and for the $\left \lceil{\frac{N \cdot m}{T}}\right \rceil$-th level each radar chooses its 
		most accurate target that has not been selected by others, where $\lceil{\cdot}\rceil$ is the ceiling function.
	\end{itemize} 
\end{prp}

\begin{proof}
Similar arguments hold as for Prop.~\ref{b_slucaj}. 
Yet, 
the game above seems to be a 
ranked anti-coordination game.
Note that here there are still multiple NE, but not all NE are necessarily equal, and hence, not every NE is Pareto optimal (only one is). So, the conditions above are not sufficient to have also a Pareto optimal NE, and consequently, PoA is strictly greater than 1.
\end{proof}


\subsection{BRD-based distributed track selection algorithm}
\label{BRD_subsection}

	In the sequel, we present a simple, low-complexity, distributed algorithm, based on the best-response dynamics (BRD)~\cite{shoham2008multiagent},
	\cite{Han_2012_book}, \cite{bacci_2015game} literature, that looks for  the NE of the analyzed game.  
	Toward this goal, 
	let us first define the notion of radar $i$'s best response to the vector of strategies $s_{-i}$, denoted by $\mathrm{BR}_i (s_{-i})$, as the set-valued function 
	$$ \mathrm{BR}_i (s_{-i}) =  \mathrm{arg} \, \max\limits_{s_i \in \mathcal{S}_i} u_i (s_i, s_{-i}) .$$
	
	Note that there are two 
	versions of BRD that can be used; namely, the sequential version 
	\begin{align*} 
s_i(k+1)  \in   \mathrm{BR}_i \big(s_1(k+1), \ldots \hspace{4cm}\\
\ldots , s_{i-1}(k+1), s_{i+1}(k), \ldots , s_N(k) \big) ,
	\end{align*}
	where $s_i(k+1)$ is the action selected by radar $i$ at time step $(k+1)$,
	and the simultaneous one 
	where all players update their actions
	synchronously
	$$ s_i(k+1) \in \mathrm{BR}_i \left(s_{-i}(k)\right) .$$
	
	Although 
	the former one is more frequently used~\cite{bacci_2015game}, it requires the definition of a cyclic path that covers all nodes, which is an NP-hard problem~\cite{papadimitriou2003computational},~\cite{6882254_TSPinc}, 
	and furthermore it has limited applicability in large and delay-intolerant networks if the whole cycle has to be performed at each time instant. 
	Thus, we focus on the simultaneous BRD implementation which, on the other hand, %
	may experience the problem of a coordination failure 
	due to strategic uncertainty %
		(see Fig.~\ref{fig:radar_coord_failure}). %
	Nevertheless, this problem can be alleviated  if radars select their best responses with some probability $\alpha<1$.

	\begin{figure}[!t]
		\centering 
		
		\subfigure[]{
			\includegraphics[width=0.24\textwidth]{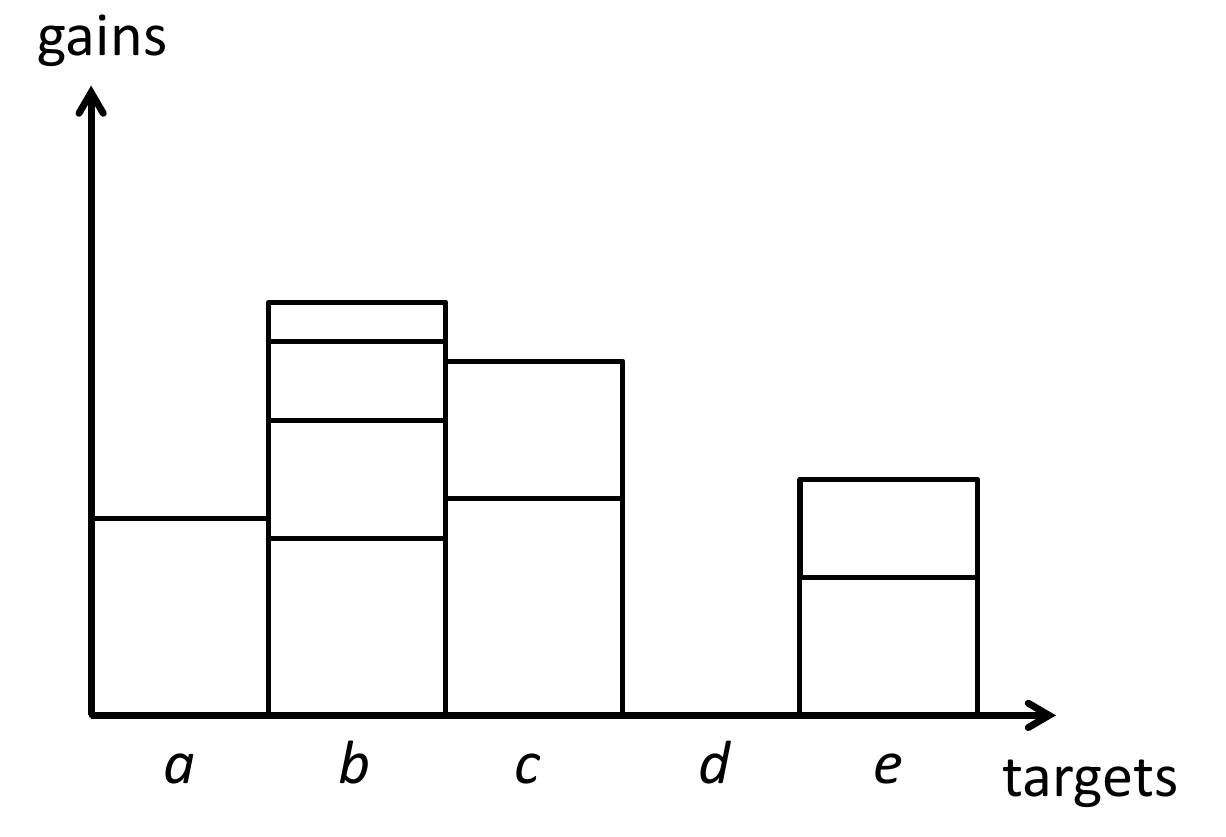}
					\label{fig:fail_a}
		} \hspace{0.05cm}
		\subfigure[]{
			\includegraphics[width=0.24\textwidth]{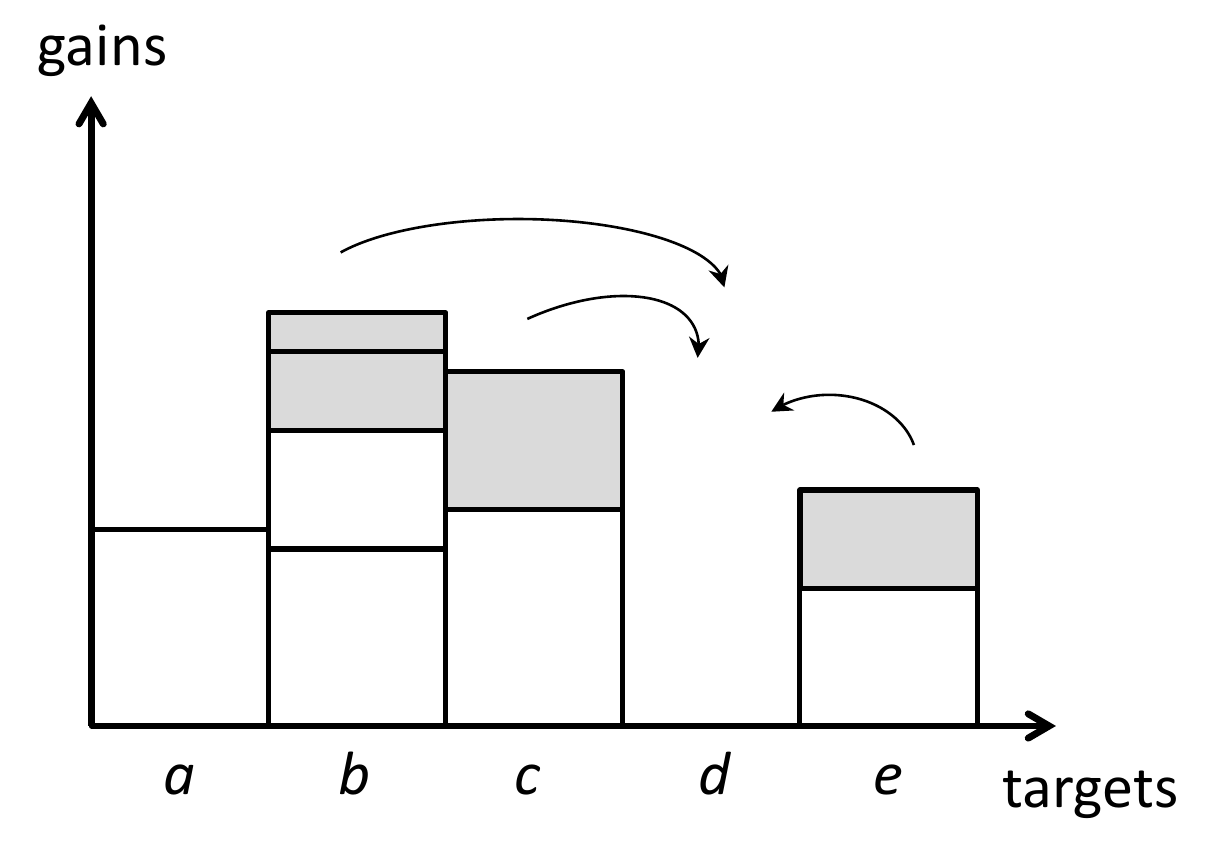}
			
					\label{fig:fail_b}
		} \hspace{0.05cm}
		\subfigure[]{
			\includegraphics[width=0.24\textwidth]{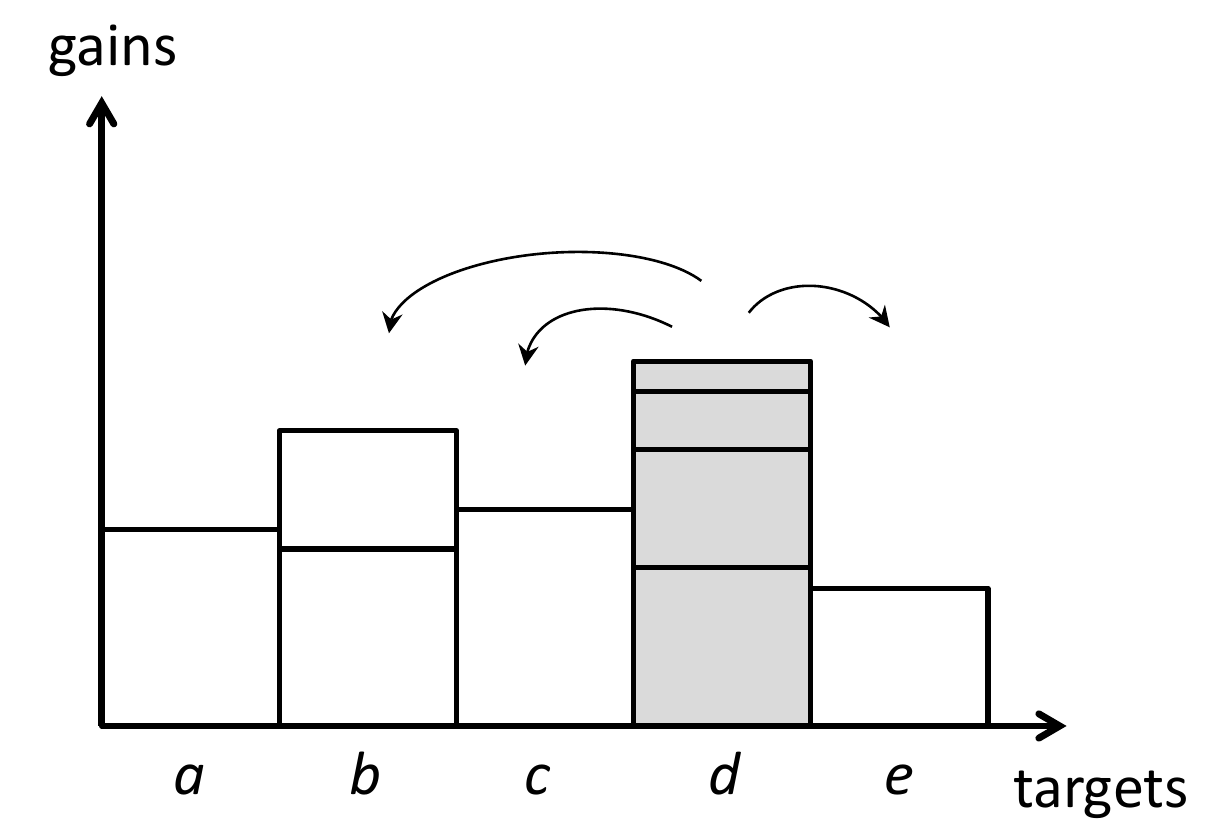}
					\label{fig:fail_c}
		}
		\caption{A coordination failure example;~\subref{fig:fail_a} initial track allocation;~\subref{fig:fail_b} four radars decide to change their current track choices (gray boxes) and illuminate target $d$;~\subref{fig:fail_c} the track allocation in the following time instant and possible radar choices (denoted by the arrows) that may result into the cyclic behavior.}
		\label{fig:radar_coord_failure}
	\end{figure}

	In the games above where $N \cdot m > T$, in general, two types of NE may arise,  
	one where a radar illuminates only different targets and the other where it chooses the same target more than once.
	In practice, it is of interest to exploit the radars' diversity; thus, we focus on the former type.
	%
	%
	%
	Let $\mathcal{T}_{\mathrm{sel}}^{(i)}$ denote the set of targets selected by radar $i$.

		Then, a summary of the proposed algorithm is provided in the following.

	\noindent
	\rule{\linewidth}{0.5mm} \\[-0.5mm]
	\textbf{Algorithm 1: Low-complexity BRD-based  distributed scheme (LC-BRD) for track selection}\\[-2mm]
	\rule{\linewidth}{0.5mm}
	\begin{itemize}

		\item 	
		Start with any strategy profile $s(0)$.
		\item At each time instant $k=1, 2, \ldots,$ each radar $i \in \mathcal{N}$ performs the following steps:
		\begin{itemize}
			\item[s1)] Count $m_j^t$, $\forall j \in \mathcal{T}$, and reallocate the measurements for $\forall j \in \mathcal{T}_{\mathrm{sel}}^{(i)}$ 
			satisfying $s_{i,j}>1$ to a target $\mathrm{ arg min}_{\ell \in \{\mathcal{T}\setminus\mathcal{T}_{\mathrm{sel}}^{(i)}\}} m_{\ell}^t$.
			\smallskip
			\item[s2)] With probability 
			$\alpha$, reallocate the measurement from target $j$ to $\ell$ until
			\begin{itemize}
				\item $\exists j \in \mathcal{T}_{\mathrm{sel}}^{(i)}$ such that $m_{j}^t > \left \lceil{\frac{N \cdot m}{T}}\right \rceil$ and the measurement for $\ell$ is the most accurate one of those satisfying $\mathrm{ arg min}_{q \in \{\mathcal{T}\setminus\mathcal{T}_{\mathrm{sel}}^{(i)}\}} m_{q}^t$, or
				\item  $m_j^t- m_{\ell}^t=1$, where $ m_j^t=\mathrm{max}_{q \in\mathcal{T}_{\mathrm{sel}}^{(i)}} m_q^t$ and $m_{\ell}^t=\mathrm{min}_{q \in \{\mathcal{T}\setminus\mathcal{T}_{\mathrm{sel}}^{(i)}\}} m_q^t$, and if measurement for $\ell$ is more accurate than the one for $j$. 
			\end{itemize} 
			\item[s3)] Transmit/receive measurements, and $\forall j \in\mathcal{T}$, execute~\eqref{eq_predict_1}-\eqref{eq_predict_2} and employ all available measurements in~\eqref{eq_update_1}-\eqref{eq_update_3}.
			
			\item[s4)] (optional) if $u_i \big(s_i(k), s_{-i}(k)\big) < u_i \big(s_i(k-1), s_{-i}(k-1)\big)$ revert back to the strategy from $k-1$ and skip the first 2 steps, i.e., $s_i(k+1)=s_i(k-1)$.
		\end{itemize}

	\end{itemize}
	\vspace{-0.25cm}
	\rule{\linewidth}{0.5mm}\\[-2mm]

	In the context of general BRD algorithms, players need to observe the actions
	played by the others; however, in our algorithm, it can be verified that the knowledge of the numbers of transmissions per each target $j$, i.e., $\left\{m^t_j\right\}_{j=1}^{T}$, is sufficient. Specifically, note that  $\left\{m^t_j\right\}_{j=1}^{T}$ are aggregate functions of the radars' actions and, due to  $u_i (s_i, s_{-i}) =   u_i (m^t_1,\ldots, m^t_{T})$,  observing the actions themselves 	is not necessary. 

Note that there are no convergence results for general games using BRD, i.e., a BRD-based algorithm may miss an NE
~\cite{shoham2008multiagent}, \cite{bacci_2015game}.
Fortunately, for some special classes of games there exist sufficient conditions under which the convergence of the sequential BRD to a pure NE is always guaranteed. %
For instance, one such class is related to the so-called potential games~\cite{Monderer1996124}, which we define next.

	\begin{defn}
		A finite, $N$-person normal-form game $\Gamma=(\mathcal{N}, \mathcal{S}, u)$ is called a \textit{potential} game\footnote{Strictly speaking, the game defined in the definition above is formally known as \textit{exact potential game}. There are other variants of potential games, where probably the most general one is the so-called \textit{ordinal potential game} in which the condition $u_i (s_i, s_{-i}) -  u_i (s^{\prime}_i, s_{-i})>0$ iff $\Phi (s_i, s_{-i}) - \Phi (s^{\prime}_i, s_{-i})>0$ holds. Most importantly, both types of potential
			games are still guaranteed to have pure-strategy Nash equilibria.} if there exists a function $\Phi \colon \mathcal{S} \rightarrow \mathbb{R}$ such that $\forall i \in \mathcal{N}$ and for all $(s_i, s_{-i}),(s^{\prime}_i, s_{-i}) \in \mathcal{S}:$
		$$u_i (s_i, s_{-i}) -  u_i (s^{\prime}_i, s_{-i}) = \Phi (s_i, s_{-i}) - \Phi (s^{\prime}_i, s_{-i}) ,$$
		and such a function $\Phi$ is called 
		potential function of the game.
\label{potential}
	\end{defn}
In every finite  potential game, every improvement path is finite. Since a finite game  has a finite
strategy space, the potential function takes on finitely many values and the above
sequence must terminate in finitely many steps in an equilibrium point.

Unlike the sequential BRD, there does not seem to exist general convergence results for the simultaneous BRD, yet only a few application-specific proofs\cite{bacci_2015game}. Nevertheless, the proposed Algorithm $1$ does converge to a pure NE.

	\begin{thm}
The proposed Algorithm 1, with the s4) step, does converge to a pure Nash equilibrium of the proposed one-shot track selection game $\Gamma^{(1)}=(\mathcal{N}, \mathcal{S}, u)$, defined  in Sec~\ref{subsec_4_1_GT_model}.
	\end{thm}

\begin{proof}	
 Let us first analyze a hypothetical, sequential version of the proposed algorithm. Note that one may construct a potential function $\Phi$ for the analyzed game, i.e., by setting $\Phi= u_i(s)$,  $\forall i \in \mathcal{N}$. Thus, a sequential BRD-based strategy for the analyzed game would converge. Now, for the proposed (simultaneous) algorithm, note that in general case $\Phi$ 
is not  non-decreasing  
as time progresses; however, 
due to the s4) step, only the states where $\Phi$ is not smaller than the best previous $\Phi$ value are actually kept. Specifically,  in case where the players at time $k$ select a coordination failure profile which may result in $\Phi(k)<\Phi(k-1)$ (such as one given in Fig.~\ref{fig:radar_coord_failure}), this step ensures that $\Phi(k+1)=\Phi(k-1)$. Then, due to $\alpha<1$, there is a  non-negligible probability that only one player will update (as in the asynchronous version) and $\Phi$ will increase; thus, the algorithm will eventually converge. 
\end{proof}

\textit{Remark 1}: Strictly speaking, the proposed algorithm with the s4) step is not a traditional simultaneous BRD, since it requires that each player also stores in memory 
the action and the utility value from the previous time step. This additional, yet small memory requirement is sufficient but not necessary for the algorithm to converge. For properly set $\alpha$, our simulations have shown that the proposed algorithm, even without the s4) step, actually converges and performs well (see also~\cite{track_selection_7472254}).

\textit{Dynamic scenario}:
Note that the tracking accuracy gain in~\eqref{eps:gain_eq}, which constitutes the utility of each radar in~\eqref{eps:GT_model_eq1}, 
generally depends on measurement noise covariance $R_{j,i}$, deterministic target dynamics $F$ and process noise covariance $Q$. 
To account for time-varying accuracy measures, i.e., range and azimuth variances, and to deal with 
possibly high target dynamics, 
the proposed algorithm can be modified in one of the following ways:
\begin{itemize}
	\item
	 LC-BRD algorithm can be 
	 repeated every $K$ time instants so as to search for other NE during the tracking process,
	%
	or
	%
	\item
	Each radar running LC-BRD  may randomly change its strategy in step $s2$ (regardless of the conditions in this step) with a small $\epsilon$ probability. In other words, step $s2$ in LC-BRD is run with probability $1-\epsilon$.
\end{itemize}
The modifications above achieve similar performance, as it will be shown in the simulation section.

\section{Scenario 2: Coordination $\&$ Conflict}
\label{coord_CONflict}

In the previous section, we have analyzed the scenario where all radars share the same interests; thus, the main challenge has been to tackle the problem of coordination among radars. %
Note that, in practice, different radars or groups of radars (or more precisely, radar operators) can have different interests over targets. For instance, the radar operators can have different target priorities. 
Or even simpler, radar operators can be interested only in a specific region and/or in a specific type of targets. Yet, in such situations it is still important to exploit the network cooperation, as in the scenario analyzed in the previous section. 
Thus, here %
we focus on a more demanding scenario where radars may have different interests and where issues of conflict may also arise. Specifically, we assume that radars:
\begin{itemize}
	\item[\textit{i)}] do not necessarily have the same target interests,
	\item[\textit{ii)}] are limited to partial target observability, and
	\item[\textit{iii)}] do not communicate with all other radars.
\end{itemize}
An example of such a scenario is depicted in Fig.~\ref{fig:radar_network_2_3_scen_2}. For instance,  radar $1$ in Fig.~\ref{fig:radar_network_2_3_scen_2}, denoted as $R_1$,  communicates with only two neighboring radars ($R_2$ and $R_3$), observes only two targets ($T_1$ and $T_2$), while being interested in tracking three targets ($T_1$-$T_3$), i.e., there are three non-zero weights, which correspond to $T_1$-$T_3$, in its weight vector $w_1$. On the other hand, $R_3$ has different yet overlapped interests and different neighbors and observability conditions.

\begin{figure}
	\centering
	\includegraphics[width=0.85\linewidth]{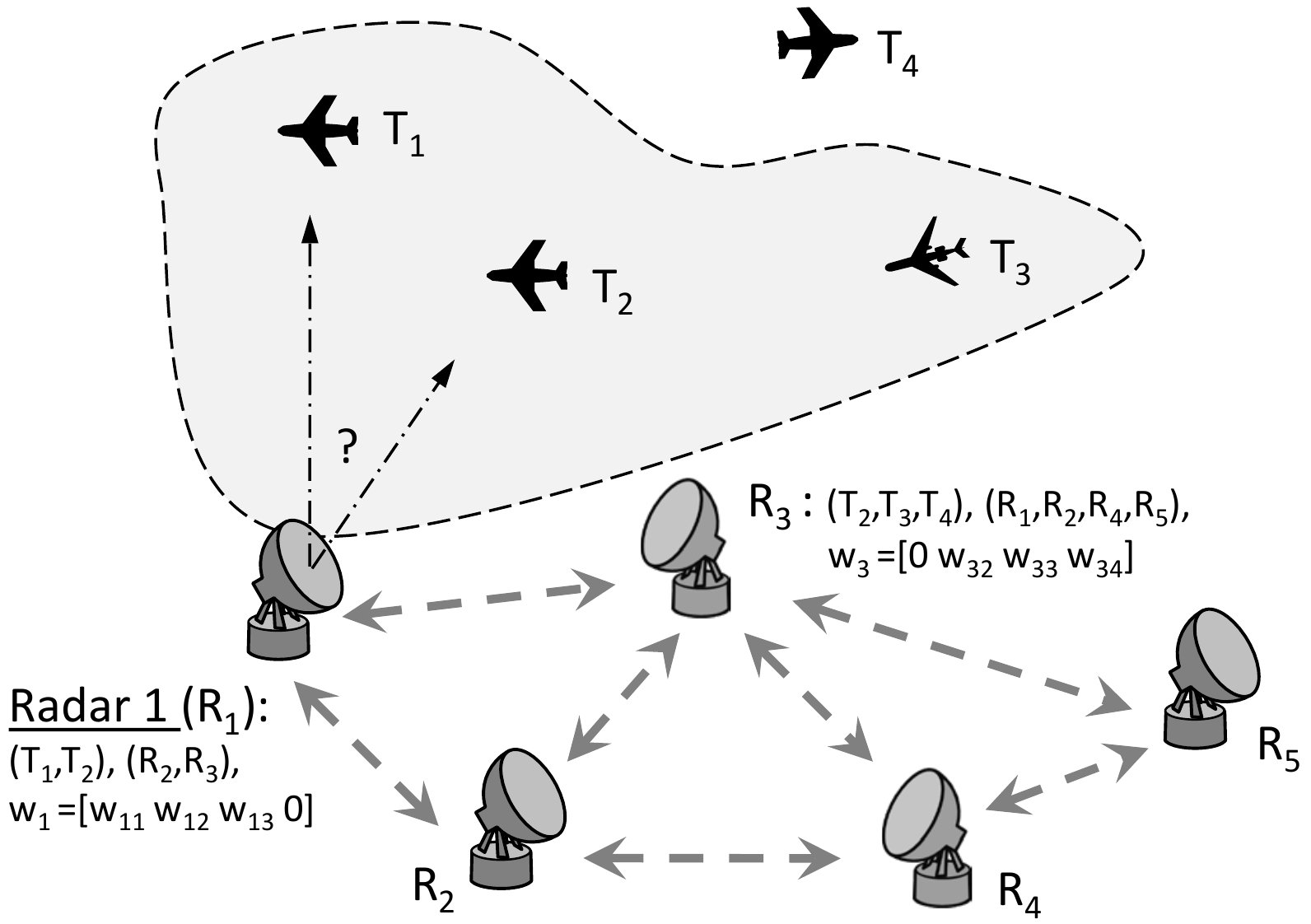}
	\caption{A track selection problem in a scenario with partial observability, limited communication and different interests among radars.}
	\label{fig:radar_network_2_3_scen_2}
\end{figure}

\subsection{Game-theoretic model}
\label{subsec_5_1_GT_redefinition}

Here, we redefine the track selection game $\Gamma^{(2)}=(\mathcal{N}, \mathcal{S}, u)$ as: 
\begin{itemize}
	\item The players are the radars represented by the set $\mathcal{N}$.
	\item The strategy of each radar $i$ is represented by a 
	$T$-tuple $s_i =(s_{i,1}, s_{i,2}, \ldots , s_{i,T})$ where 
	\begin{align}
	s_{i,j} =\begin{cases}
	a,& \text{if } j \in \mathcal{T}_i\\
	0,              & \text{otherwise}
	\end{cases}
	\end{align}
	where $a$ is the number of transmission beams that radar $i$ devotes to target $j \in \mathcal{T}_i$ and it holds that $a\leq m$.
	Each strategy-tuple has at most $m$ transmissions, i.e., $\sum_{j \in \mathcal{T}_i} s_{i,j} \leq m$. 
Now, the number of transmissions each radar $i$ collects from 
$\mathcal{N}_i$ and 
related to some target $j \in \bigcup_{i \in \mathcal{N}_i}\mathcal{T}_i $ is given as $m^t_j (i)=\sum_{i \in \mathcal{N}_i} s_{i,j}$.
	\item The utility for each radar $i$ is given by
	\begin{gather}\label{eps:GT_model_eq1_other}
		\begin{split}
			u_i (s_i, s_{-i}) =     \sum_{j=1}^{T}  w_{i,j} \cdot \mathrm{gain}_j\left(m^t_j(i)\right) \, , 
		\end{split}
	\end{gather}
with $\mathrm{gain}_j\left(m^t_j(i)\right) =\mathrm{Tr} \left\{ P_{j, k\vert k-1}  - P_{j, k\vert k}^{(m^t_j(i))} \right\}$, and $w_{i,j}$ being the weight that a radar $i$ gives to some target $j$. In fact, $w_{i,j}$ can be seen as 
$(i,j)^\mathrm{th}$ element of $N \times T$ matrix $W$ which defines the target interests across all radars. Also, note that $u_i (s_i, s_{-i})=u_i \left( \{s_l\}_{l \in \mathcal{N}_i} \right)$.

\end{itemize}

\subsection{Correlated equilibria and regret-matching}
\label{subsec_5_2_CE}

The scenario considered in the previous example resembles the well-known \textit{Battle of the Sexes} game~\cite{shoham2008multiagent} where players have a common interest to coordinate (or in our case to anti-coordinate), but they have different preferences regarding the (anti-)coordinated states of the game (which are NE). 
However, for a general setting of our game defined above, 
it is not easy to characterize possible NE neither in terms of their efficiency nor even their existence. Furthermore, we cannot ensure that the game is potential. This is due to the fact that, in general case, it is difficult to construct a potential function $\Phi$ since 
an action profile change can influence different players in an arbitrarily different way (see an example in Fig.~\ref{fig:poten}).

\begin{figure}[!t]
	\centering 
	\subfigure[$R_1$, time $k$]{
		\includegraphics[width=0.2\textwidth]{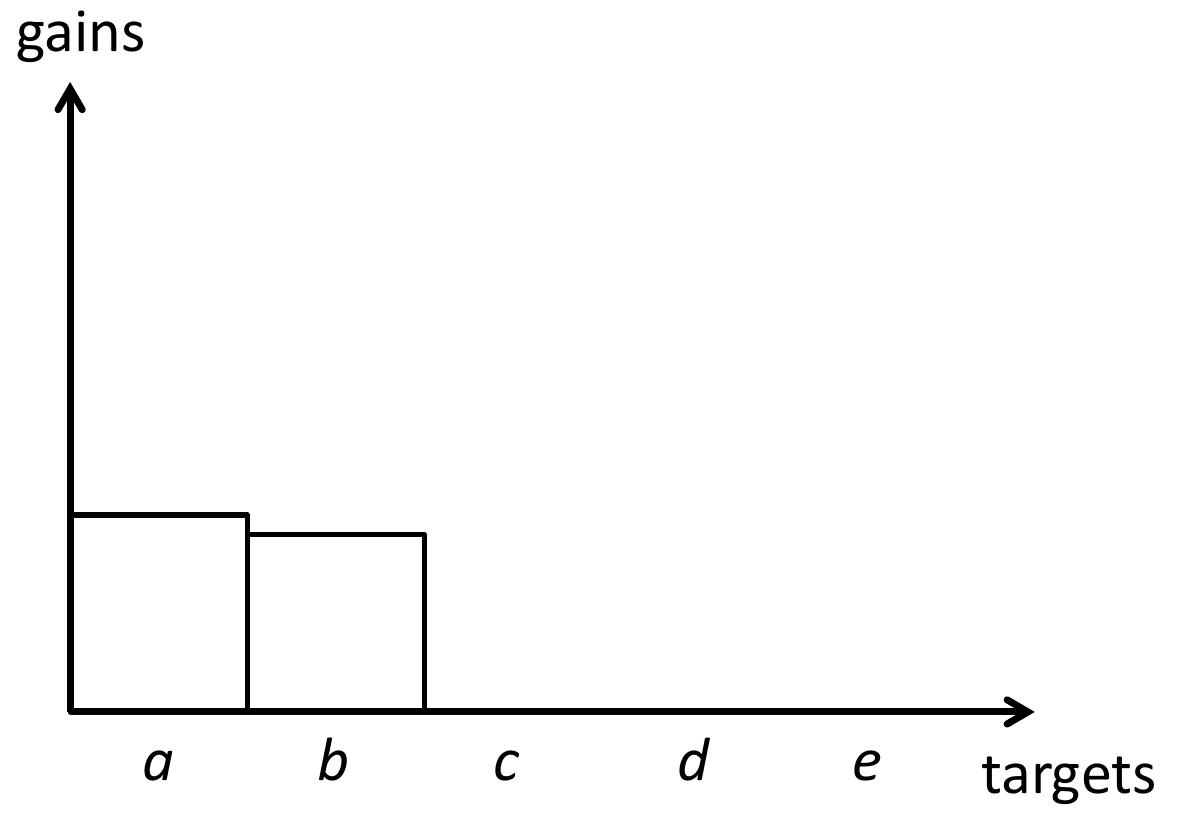}
		\label{fig:poten_1}
	} \hspace{0.1cm}
	\subfigure[$R_3$, time $k$]{
		\includegraphics[width=0.2\textwidth]{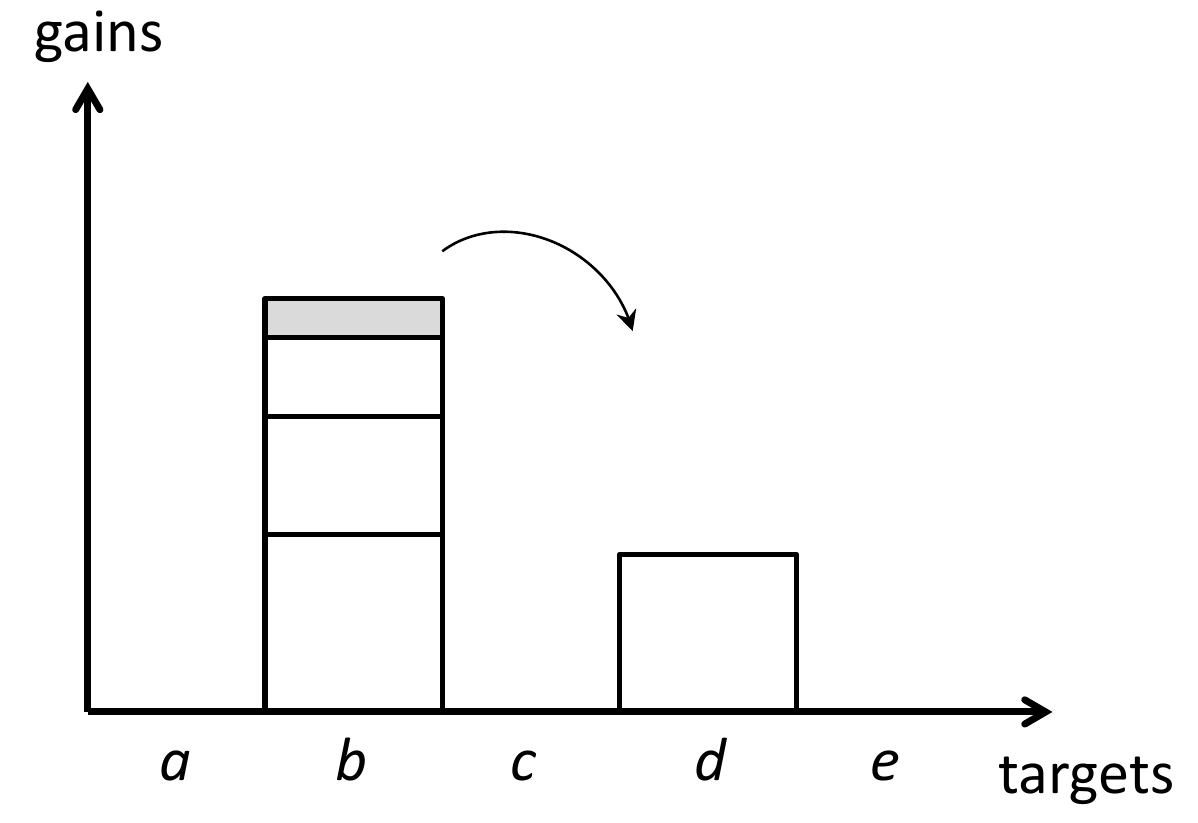}
		\label{fig:poten_2}
	}
		\subfigure[$R_1$, time $k+1$]{
			\includegraphics[width=0.2\textwidth]{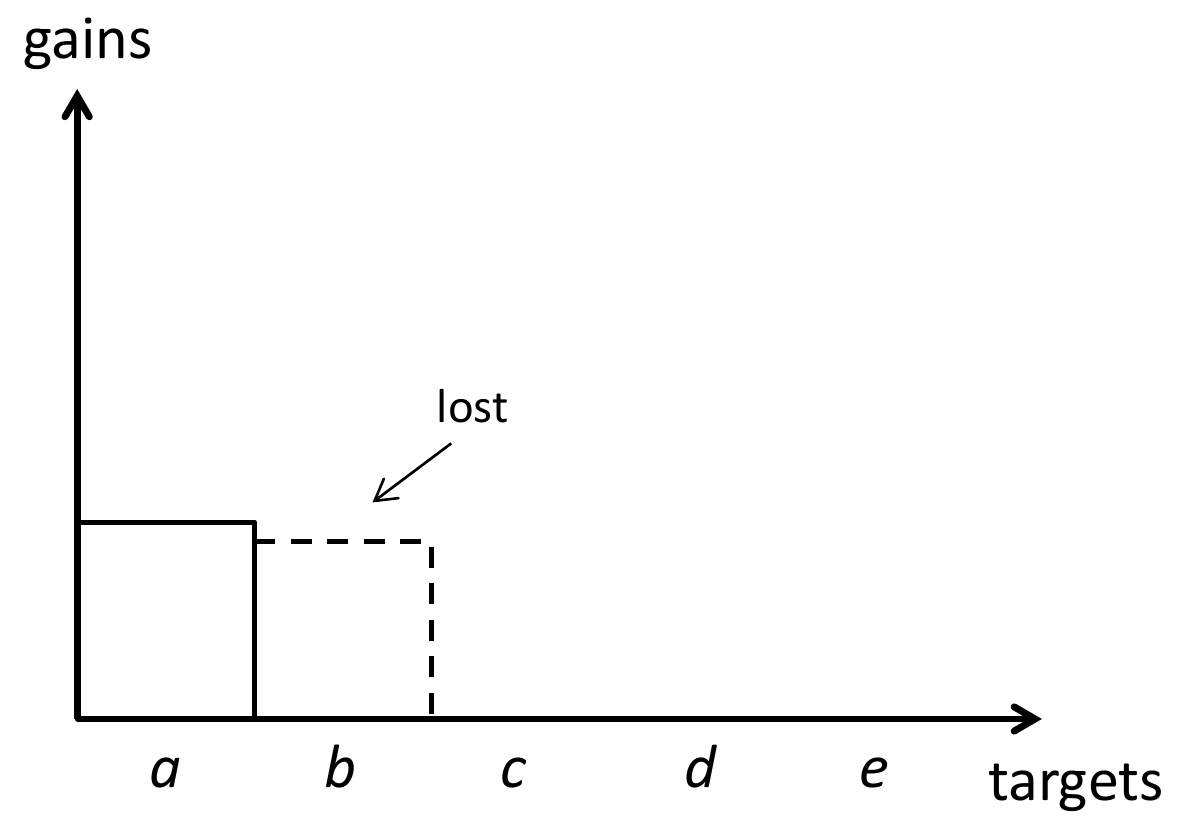}
			\label{fig:poten_3}
		} \hspace{0.1cm}
		\subfigure[$R_3$, time $k+1$]{
			\includegraphics[width=0.2\textwidth]{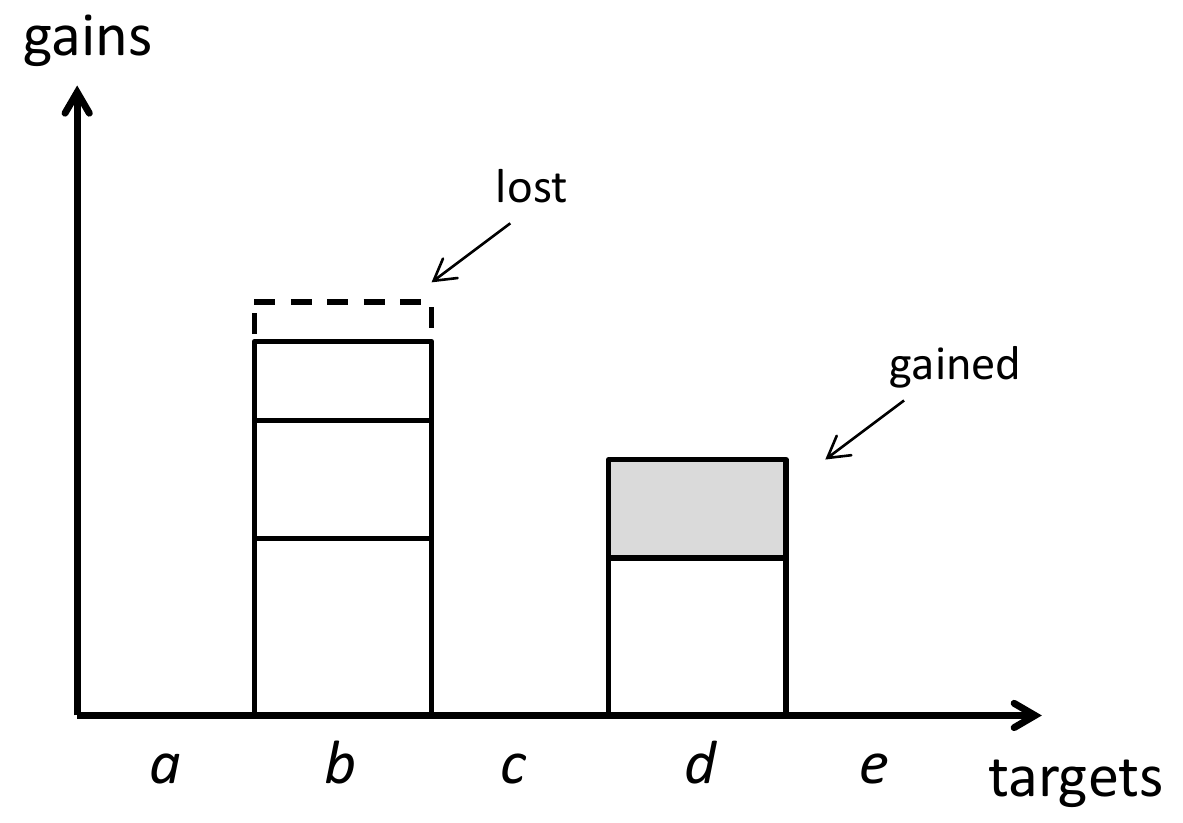}
			\label{fig:poten_4}
		}
	\caption{Assume a track allocation across the network as in Fig.~\ref{fig:example1b}. Due to limited connectivity and different interests, at time $k$, radar 1 ($R_1$) experiences the gains given in~\subref{fig:poten_1}, and those in~\subref{fig:poten_2} are for radar 3 ($R_3$) which decides to change its selection. At time $k+1$, $R_1$ has a (great) loss while $R_3$ has a (relatively small) gain.  
	}  
	\label{fig:poten}
\end{figure}

\textit{Remark 2}: In the extreme case where radars have totally different interests (with no overlap w.r.t. the interests and communication topology), then it would be easy to define a potential function (just the sum of all utilities). However, the solution (NE) is trivial since the problem is totally decoupled (there is no inter-dependence); each radar's utility depends only on its own strategy selection $u_i(s)=u_i(s_i), \forall i \in \mathcal{N}$.

For the reasons mentioned above, here we focus on the solution concept of correlated equilibrium (CE). 
Note first that, as mentioned in Sec.~\ref{Sec_2_GT_overview}, a CE always exists in a finite game~\cite{hart2013simple_book}. Actually, every NE is a CE 
and NE correspond to the special case of a CE for which the joint distribution over the strategy profiles $\pi(s_{i}, s_{-i})$ factorizes as the product of its marginals, 
i.e., the play of different players is independent~\cite{Aumann197467}, \cite{hart2013simple_book}. Furthermore, in certain settings the set of CE may include  even the distribution that is not in the convex hull of the NE distributions.

Next, we will exploit a  class of simple, adaptive algorithms, called \textit{regret matching}
, in order to reach a CE of the analyzed track selection game. It
does not entail any sophisticated updating, prediction, or fully rational behavior~\cite{Hart_paper_a_simple}. The approach can be summarized as follows: 
At each time instant, a radar may either continue
playing the same target strategy as in the previous time instant, or switch to other strategies,
with probabilities that are proportional to how much higher his accumulated
accuracy gain would have been had he always made that change in the past. Specifically, 
at each time instant $k$ and for any two distinct strategies $s_i^{\prime}\neq s_i$, the regret which radar $i$ experiences at time $k$ for not playing $s_i^{\prime}$ is given by
	\begin{gather}\label{eps:RM_eq1}
	\begin{split}
	R_{i, k}(s_i, s_i^{\prime}) =\max\{D_{i, k}(s_i, s_i^{\prime}), 0\},
\end{split}
	\end{gather}
where the term $D_{i, k}(s_i, s_i^{\prime})$ represents the average payoff at time $k$ for not having played, every time that $s_i$ was played in the past $k_p \leq k$, the different strategy $s_i^{\prime}$:
	\begin{gather}\label{eps:RM_eq2}
	\begin{split}
	& D_{i, k}(s_i, s_i^{\prime})=\\&  \qquad ={1\over k}\sum\limits_{k_p\leq k}\left[u_i \left(s_i^{\prime}(k_p), s_{-i}(k_p)\right) - u_i \left(s_i(k_p), s_{-i}(k_p)\right)\right].
	\end{split}
	\end{gather}
Next, the probability at time $k+1$ for radar $i$ to play some strategy $s_i^{\prime} \in \mathcal{S}_i$ is a linear function of its regret vector, i.e.,
	\begin{gather}\label{eps:RM_eq_probabilities}
	\begin{split}
 \begin{cases} \pi^{k+1}_i(s_i^{\prime})= \frac{1}{\mu} R_{i, k}(s_i, s_i^{\prime}) ,\quad \text{for all} \,\, s_i^{\prime}\neq s_i,  \\ 
\pi^{k+1}_i(s_i)  =1-\sum\nolimits_{s_i^{\prime}\neq s_i} \pi^{k+1}_i(s_i^{\prime}),\quad \text{otherwise}, \end{cases}
	\end{split}
	\end{gather}
	where the fixed constant $\mu >0$ is selected to be large enough such that  $\pi^{k+1}_i(s_i)>0$.
Finally, for every $k$, we define the empirical distribution $\eta_{k}$ of the strategy profiles played up to time $k$, i.e., for each $s \in \mathcal{S}$, 
	\begin{gather}\label{eps:RM_eq3}
	\begin{split}
	\eta_{k}( s)={1\over k} \#\{k_p\leq k: s(k_p)=s\}, 
	\end{split}
	\end{gather}
with $\#(\cdot)$ standing for the number of times the event inside the brackets occurs while $s(k_p)$ is the action profile played at time $k_p$.  

\begin{thm}
 If every radar select targets according to \eqref{eps:RM_eq_probabilities}, then the empirical distributions $\eta_{k}$ converge almost surely as $k \rightarrow \infty$  to the set of correlated equilibrium distributions of the game $\Gamma^{(2)}$.
\end{thm}
\begin{proof}
For the proof, see \cite{Hart_paper_a_simple}.
\end{proof}  

\textit{Dynamic scenario}: 
Due to possibly time-varying accuracy measures and high target dynamics, as explained in the end of Sec.~\ref{BRD_subsection}, as well as time-varying radar interests, the suggested approach has to be modified so as to take into account the aforementioned effects.  In fact, by incorporating an adaptive mechanism in the calculation of the average regret, it can be shown that the resulting algorithm can track the changes if they are sufficiently small. 

Firstly, note that the average regret in~\eqref{eps:RM_eq2} can be computed recursively, i.e., 
		\begin{gather}\label{eps:RM_eq_recur1}
		\begin{split}
		D_{i, k}(s_i, s_i^{\prime})= &{k-1\over k} \,	D_{i, k-1}(s_i, s_i^{\prime}) +\\ & +  {1\over k}\left[u_i \left(s_i^{\prime}(k), s_{-i}(k)\right)- u_i \left(s_i(k), s_{-i}(k)\right)\right].
		\end{split}
		\end{gather}
Also, the average regret in~\eqref{eps:RM_eq2} and~\eqref{eps:RM_eq_recur1} exploits the history of all past selections. This is not desirable due to the fact that the tracking accuracy gains slightly change in time due to the 
aforementioned effects. Thus, to compute the average regret, each radar should exponentially discount the influence of its past selections. Specifically, similarly to \cite{4784355_Maskery_Krishna}, we rewrite the average regret recursion as:
		\begin{gather}\label{eps:RM_eq_recur2}
		\begin{split}
		D_{i, k}(s_i, s_i^{\prime})=& D_{i, k-1}(s_i, s_i^{\prime})  \\
		& +  \theta_k \Big[u_i \left(s_i^{\prime}(k), s_{-i}(k)\right) - u_i \left(s_i(k), s_{-i}(k)\right) \\& \hspace{3.8cm} - D_{i, k-1}(s_i, s_i^{\prime})\Big].
		\end{split}
		\end{gather}
where $\theta_k$ is a positive step-size. In case where the step-size $\theta_k$ is decreasing with time, the algorithm will converge with probability 1 to the correlated equilibria of a static game. In fact, if $\theta_k={1\over k} $, then the recursions in~\eqref{eps:RM_eq_recur1} and~\eqref{eps:RM_eq_recur2} are identical; thus, the convergence arguments from~\cite{Hart_paper_a_simple} directly apply. However, for the decreasing step-size, the algorithm may not adapt to the changes caused by the target dynamics. On the other hand, with the fixed step-size $\theta_k=\theta$, the algorithm is able to adapt to the changes and can be proved to converge to the set of CE by using the arguments from stochastic averaging theory~\cite{kushner2003stochastic}. For a more detailed discussion, see \cite{4784355_Maskery_Krishna}.

Finally, we provide the algorithm based on regret-matching. 

	\noindent
	\rule{\linewidth}{0.5mm} \\[-0.5mm]
	\textbf{Algorithm 2: Regret-matching  distributed scheme (RM) for track selection}\\[-2mm]
	\rule{\linewidth}{0.5mm}
	\begin{itemize}

		\item 	
        Start with some initial probability vector $\pi^{1} (s)$

		\item At each time instant $k=1,2, \ldots,$ each radar $i \in \mathcal{N}$ performs the following steps:
%

		\begin{itemize}
			\item[s1)] Select target(s) according to probabilities $\pi^{k}_i(s_i^{\prime})$ and $\pi^{k}_i(s_i)$ $\forall s_i^{\prime}\neq s_i$, and denote the selection by $s_i$.
			\item[s2)] Calculate $D_{i, k}(s_i, s_i^{\prime})$	using~\eqref{eps:RM_eq_recur2}. 
			\smallskip
			\item[s3)] Calculate regret $R_{i, k}(s_i, s_i^{\prime})$ using~\eqref{eps:RM_eq1}.
			\item[s4)] 	Find probabilities for the following time instant, 
			i.e., $\pi^{k+1}_i(s_i^{\prime})$ and $\pi^{k+1}_i(s_i)$ using~\eqref{eps:RM_eq_probabilities}.
		\end{itemize}
	\end{itemize}
	\vspace{-0.25cm}
	\rule{\linewidth}{0.5mm}\\[-2mm]

\textit{Computational complexity}:
It is of interest to comment on the computational complexity of the distributed algorithms proposed in this paper. %
For illustration only, let us consider that the number of observed targets is the same for all radars, i.e., $\vert \mathcal{T}_i \vert= \vert \mathcal{T}_n \vert, \forall i,n \in \mathcal{N}$, and that the radars are interested in all targets that are observable to them.  
Then, the RM-based algorithm has the complexity that is linear in the number of radars but exponential in $m$,  i.e., $\mathcal{O}(N\cdot\vert \mathcal{T}_i \vert^m)$. This is in contrast to the centralized approach that can be realized by an exhaustive search and which has the exponential complexity also in the number of radars, i.e.,  $\mathcal{O}(\vert \mathcal{T}_i \vert^{N\cdot m})$. On the other hand, note that the LC-BRD proposed in Sec.~\ref{coordination} is the most efficient from the computational perspective; its complexity is in the order of   $\mathcal{O}(N\cdot m\cdot\vert \mathcal{T}_i \vert)$.

\section{Simulations}
\label{simulat}

In this section, we provide some computer simulations that verify the main findings and demonstrate the effectiveness of the proposed algorithms. 


Firstly, we consider an MFR network of 
$N=3$ radars, each of them making $m=2$  measurements per scan and aiming at tracking 
$T=5$ targets, see Fig.~\ref{fig:radar_network_2_3}. Specifically, the coordinates of radars are given in Table~\ref{my-label_table1}. 
%
The targets follow white noise constant velocity trajectories with initial  $x$, $y$-coordinates and velocities provided in Table~\ref{my-label_table2}. 
%
%
%
Initial guesses $\mathrm{x}_{j,0\vert0}$ are noisy versions of the initial states $\mathrm{x}_{j,0}$ and initial covariances are equal to $P_{j,0\vert0}=P_{0\vert0}=\mathrm{diag}\big\{(0.1\,km)^2,\break (0.1\,km)^2, (0.1\,km/s)^2, (0.1\,km/s)^2\big\}$. 
%
%
The update time is $t_u=0.25\,s$, and in order to model moderate maneuverability,  $\sigma^2_w$ is set to $2.5 \cdot 10^{-5} \,km^2/{s^3}$. 
%
Also, the standard deviation in azimuth is $\sigma^{(i)}_{a_j}=\sigma_a=2 \,mrad$, while the range accuracy varies among the radars and targets as $\sigma^{(i)}_{r_j}=b_{i,j}\cdot \sigma_r$, where $\sigma_r=15\,m$ and coefficient $b_{i,j}$ is taken from the interval $[1, 4.5]$.

	\begin{table}[t]
		\centering
		\caption{Radar positions}
		\label{my-label_table1}
	\small
		\begin{tabular}{l|c|c|ll}
			& \multicolumn{1}{l|}{$x$ {[}km{]}} & \multicolumn{1}{l|}{$y$ {[}km{]}} &  &  \\ \cline{1-3}
			radar 1 & -10                                     & 0                                       &  &  \\ \cline{1-3}
			radar 2 & 3                                       & 0                                       &  &  \\ \cline{1-3}
			radar 3 & 10                                      & 0                                       &  &  \\ \cline{1-3}
		\end{tabular}
	\end{table}

\begin{table}[h]
	\centering
	\caption{Target parameters}
	\label{my-label_table2}
	\small
	\begin{tabular}{l|c|c|c|c|}
		& \multicolumn{1}{l|}{$x$ {[}km{]}} & \multicolumn{1}{l|}{$y$ {[}km{]}} & \multicolumn{1}{l|}{$v_{x}$ {[}km/s{]}} & \multicolumn{1}{l|}{ $v_{y}$ {[}km/s{]}} \\ \hline
		Target 1 & 1                                      & 6                                      & 0.5                                  & 0.1                                  \\ \hline
		Target 2 & 0.5                                    & 7                                      & 0.35                                 & -0.1                                 \\ \hline
		Target 3 & 1.5                                    & 3                                      & -0.3                                 & 0                                    \\ \hline
		Target 4 & 2                                      & 4                                      & -0.2                                 & 0.1                                  \\ \hline
		Target 5 & 2.5                                    & 5                                      & 0.3                                  & 0.2                                  \\ \hline
	\end{tabular}
\end{table}

\begin{figure}[h!]
	\centering
	\includegraphics[width=0.85\linewidth]{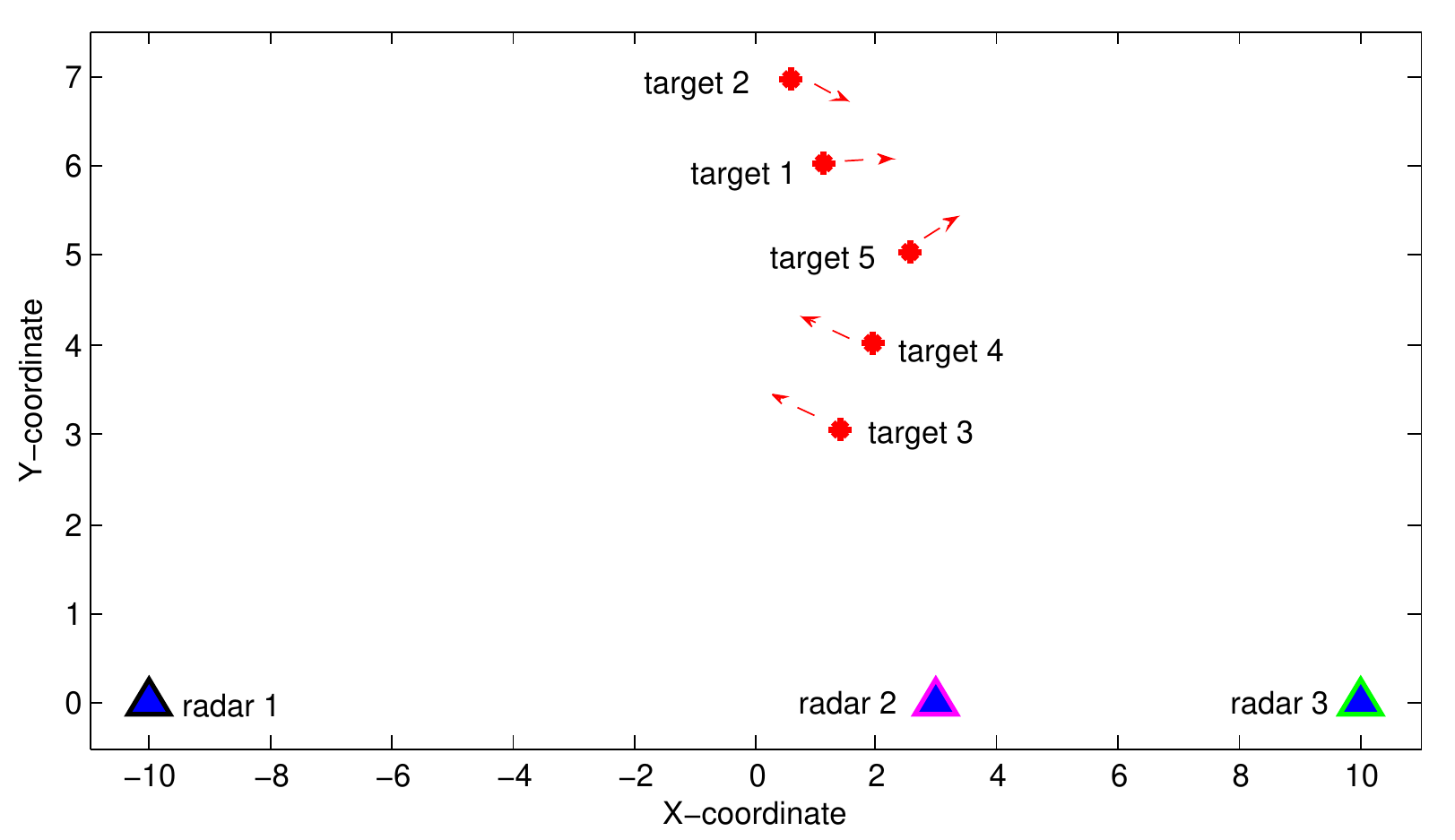}
	\caption{The coordinates of radars and targets. }
	\label{fig:radar_network_2_3}
\end{figure}

Most figures present the weighted sum of 
$\mathrm{Tr} \{P_{j, k\vert k}^{(m^t_j(i))} \}$ over all targets and over  all radars, i.e., 
\begin{gather}\label{eps:simul_eq1}
\begin{split}
\sum_{i=1}^{N} \sum_{j=1}^{T}  w_{i,j} \cdot \mathrm{Tr} \left\{P_{j, k\vert k}^{(m^t_j(i))} \right\},
\end{split}
\end{gather}
%
as a function of time $k$. Initially, we focus on the case analyzed in Sec.~\ref{coordination} %
%
%
and %
compare the following strategies:
\begin{itemize}
\item[a)] \textit{Standalone} -- The standalone radar that does not send/receive measurements. It sequentially chooses $m=2$ different targets each time instant.
\item[b)] \textit{Distributed random with $K=10$} -- Distributed strategy where the radars exchange the measurements while each of them randomly selects targets each $K=10$ time instants.
\item[c)] \textit{Distributed random with $K=1$} -- Same as in (b), except that targets are being randomly chosen at each time instant, i.e., $K=1$.
\item[d)] \textit{Proposed LC-BRD distributed with $K=10$} -- The proposed low-complexity BRD-based distributed algorithm seeking NE while being reinitialized every $K=10$ time instants. The probability $\alpha$ is set to the value of $0.5$.
\item[e)] \textit{Approximated centralized with $K=10$} -- The approximated centralized approach based on analytically resolved 
measurements-to-target allocation every $K=10$ time instants. Due to its exponential search complexity in the total number of measurements, i.e., $\mathcal{O}(T^{N \cdot m})$, the centralized exhaustive search is computationally challenging even for the considered scenario of $T= 5$, $N=3$ and $m=2$. 
For this reason, the coefficients in noise variances  $\sigma^{(i)}_{r_j}$, which are in the interval $[1, 4.5]$, are set in such a way that the best centralized measurements-to-target allocations can be easily analytically determined and changed every $K=10$ time instants.\footnote{This is only done for the purpose of a comparison. In the scenarios with limited observability, and thus less computational complexity,  we will provide the exhaustive search results as a benchmark.}
\end{itemize}

\begin{figure}
	\begin{minipage}[b]{1.0\linewidth}
		\centering
		\centerline{\includegraphics[width=7.7cm]{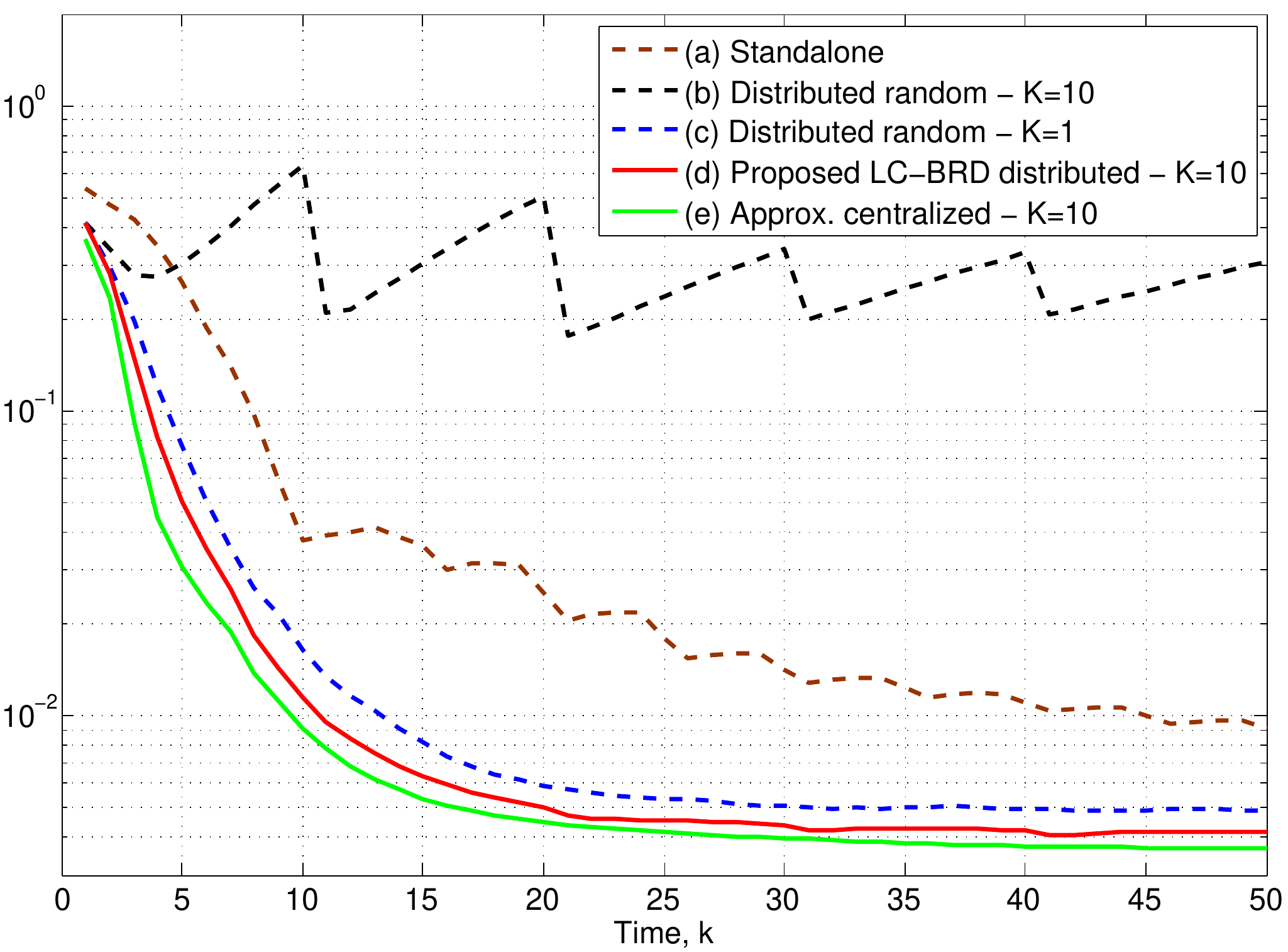}}
		
		\caption{Sum of traces of error covariance matrices for all targets during time in the setting with $T=5$ targets, $N=3$ radars, $m=2$  measurements per scan  and update time $t_u=0.25\,s$.}
		\label{fig:Results_for_P}
	\end{minipage}
\end{figure}

\begin{figure}
	\begin{minipage}[b]{1.0\linewidth}
		\centering
		\centerline{\includegraphics[width=7.7cm]{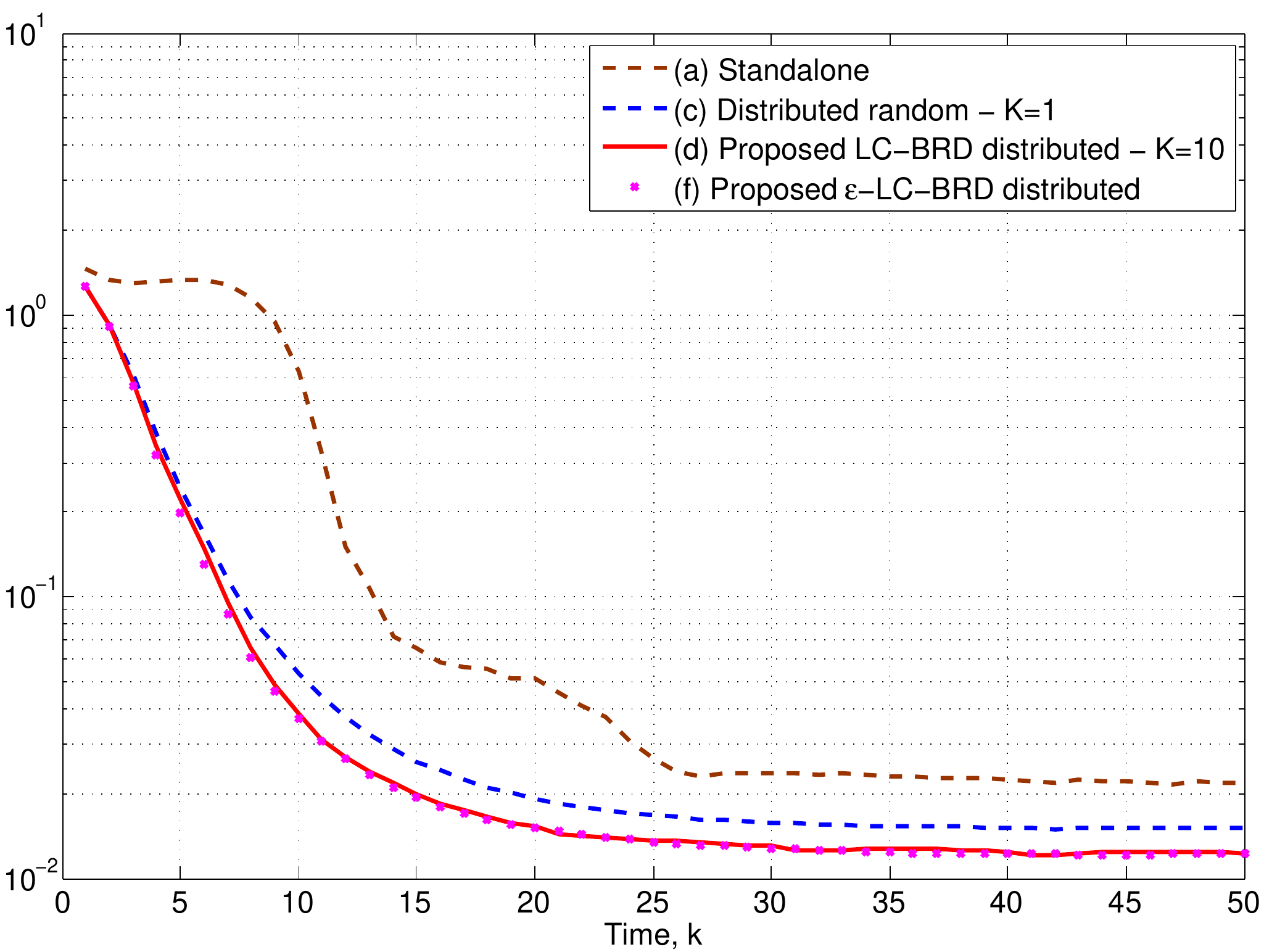}}
		
		\caption{A setting with $T=15$ targets, $N=3$ radars, $m=6$  measurements per scan  and update time $t_u=0.25\,s$.}
		\label{fig:Many_targets_many_m}
	\end{minipage}
\end{figure}

Figure~\ref{fig:Results_for_P} compares the above strategies. 
The results are averaged over 100 
realizations. 
Not surprisingly, due to the high process' dynamics, a standalone, non-cooperative radar experiences weak performance since it utilizes only its own measurements which are not sufficient to cover all targets. Although approach in (b) uses $N \cdot m =6$ measurements, due to the lack of coordination it performs poorly. However,  the distributed random strategy can be significantly improved if strategies are constantly being changed, given that there are no track migration costs involved. Note that the proposed LC-BRD distributed algorithm, which learns underlying NE allocations, outperforms the aforementioned strategies. On the other hand, it  closely approaches the performance of the approximated centralized one while 
mitigating its inherent complexity. 

The results related to a setting with more targets and measurements per scan than in the previous setting are plotted in Fig.~\ref{fig:Many_targets_many_m}. Here, we also include:
\begin{itemize}
	\item[f)] \textit{$\epsilon$-LC-BRD distributed} -- The proposed low-complexity BRD-based distributed algorithm seeking NE and where each radar may change its strategy (even if in an NE) with a small $\epsilon$ probability. The probabilities $\alpha$ and $\epsilon$ are chosen to be $0.5$ and $0.02$, respectively. 
\end{itemize}
Although without the curve for approximated centralized solution as a benchmark, the results that Figure~\ref{fig:Many_targets_many_m} provides are similar to those in Fig.~\ref{fig:Results_for_P}. Also, note that the two versions of the proposed LC-BRD algorithm exhibit pretty similar performance.

For the plots in Figs.~\ref{fig:7}-\ref{fig:mogucnost_8},  two additional radars in the network are considered w.r.t. the setting in Fig.~\ref{fig:radar_network_2_3}, i.e., $(x_4,y_4)=[-4\,km$, $0\,km]$, 	$(x_5,y_5)=[7\,km$, $0\,km]$.
This is probably the least favorable scenario for LC-BRD w.r.t. the distributed totally random ($K=1$) algorithm due to the fact that now $mod(N\cdot m, \,T)=0$, i.e., all targets can be selected with the same number of measurements. In Figs.~\ref{fig:7}-\ref{fig:mogucnost_8}, the update time is set to $t_u=0.25\,s$ and $t_u=0.025\,s$, respectively. Also, in Fig.~\ref{fig:mogucnost_8} the following strategy is used in the comparison:  
\begin{itemize}
	\item[g)] \textit{RM distributed} -- The proposed distributed algorithm based on Regret-Matching which 
	tracks CE, with $\theta_k=0.5$.
\end{itemize}
It can be noticed that the RM distributed algorithm clearly outperforms other distributed strategies. This is due to its more sophisticated learning mechanism, 
which comes at the expense of somewhat higher computational complexity than the other distributed strategies.

			\begin{figure}[!t]
				\centering 

					\includegraphics[width=0.44\textwidth]{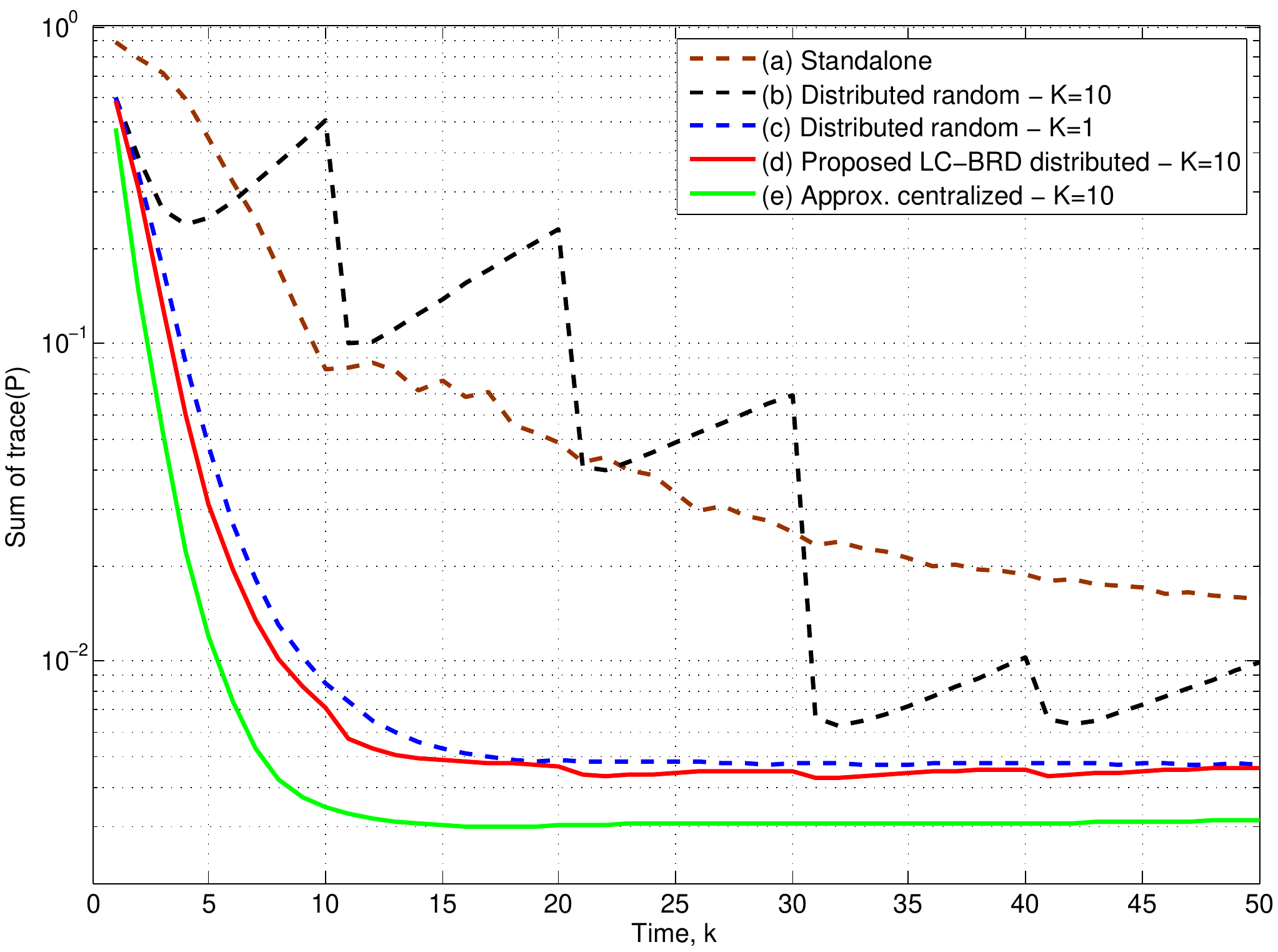}

				\caption{ 
					A setting with $T=5$ targets, $N=5$ radars, $m=2$ and update time $t_u=0.25\,s$. 
					}
				\label{fig:7}
			\end{figure}

			\begin{figure}[t!]
				\centering
				\centerline{\includegraphics[width=7.cm]{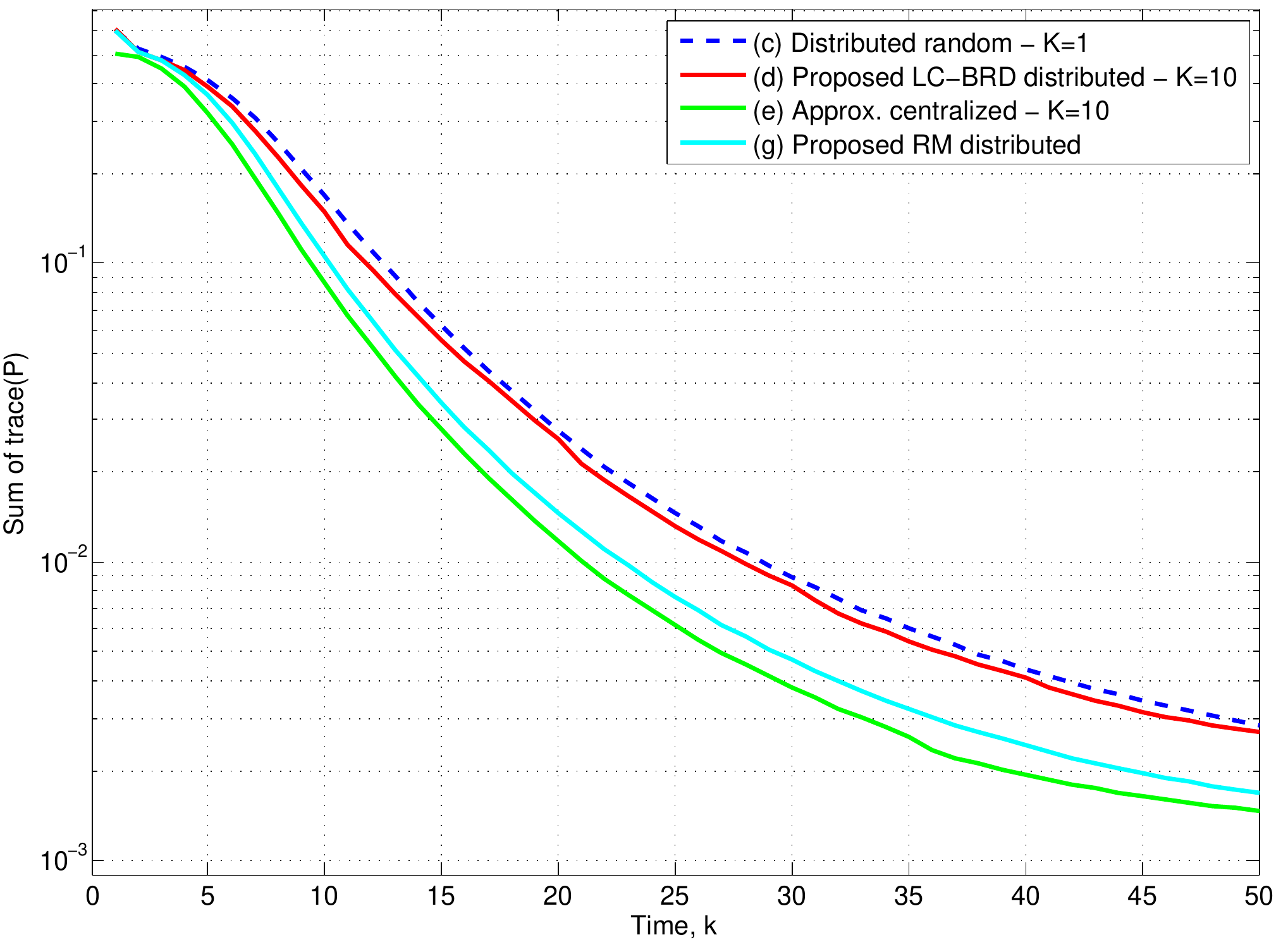}}
				\caption{Scenario is same as in Fig. \ref{fig:7} except that $t_u=0.025\,s$ and the comparison is made with the proposed distributed algorithm based on Regret-Matching.}
				\label{fig:mogucnost_8}
			\end{figure}

		\begin{figure}[h!]
			\centering
			\centerline{\includegraphics[width=6cm]{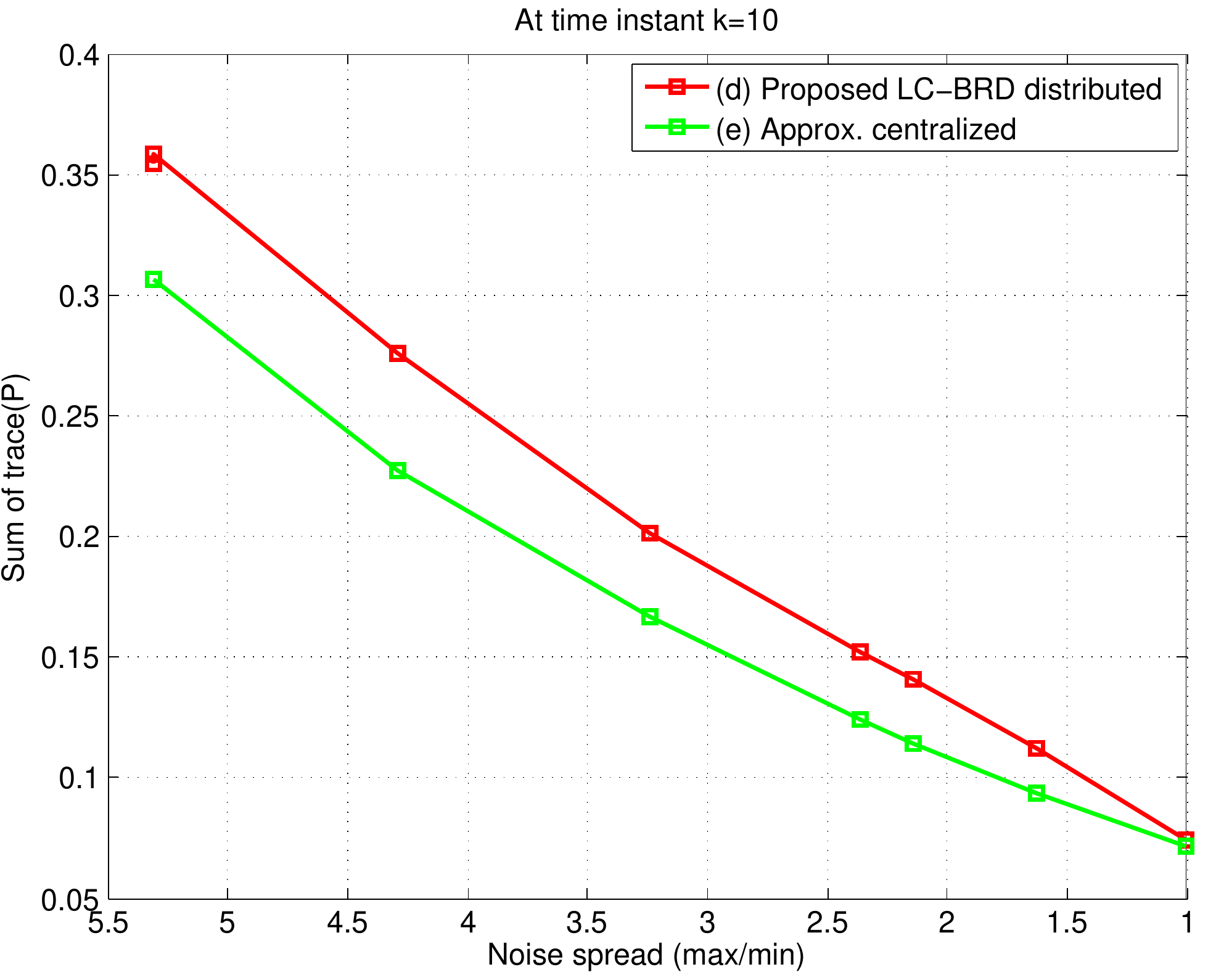}}
			\caption{Performance comparison as a function of the measurements diversity in terms of the noise variance spread, 
				for the same setting as in Fig. \ref{fig:mogucnost_8}.}
			\label{fig:diversity9}
		\end{figure}

		\begin{figure}[h!]
			\centering
								\centerline{\includegraphics[width=6cm]{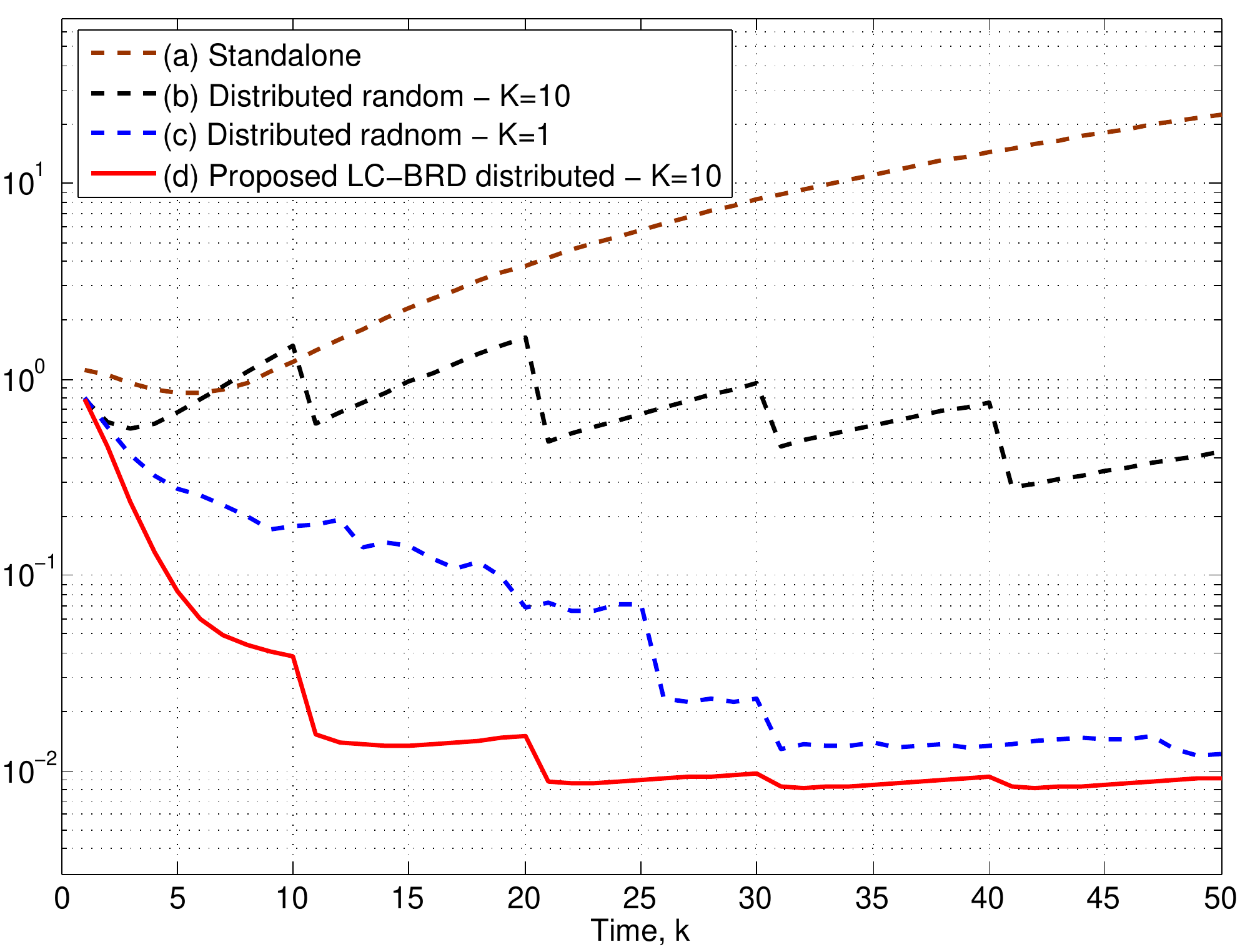}}
			\caption{Scenario with    $ 3\leq \vert\mathcal{T}_i\vert\leq 5 $, $ 5\leq \vert\mathcal{N}_i\vert\leq 6 $ and  $W=\mathrm{1}_N \cdot \mathrm{1}^T_{\vert\mathcal{T}\vert}$, where it holds $\bigcup_{i \in \mathcal{N}_i}\mathcal{T}_i =\mathcal{T}$, $\forall i \in \mathcal{N}$.}
			\label{fig:1_radar_1_target}
		\end{figure}

Regarding the LC-BRD algorithm, it should be mentioned that its performance difference w.r.t. the centralized approach mainly depends on the measurement diversity, as suggested by Props.\ref{b_slucaj}-\ref{less_case_c} in Sec.~\ref{coordination}. Specifically, the more similar measurements' quality is, the  gap w.r.t. the centralized approach is smaller ($PoA \rightarrow 1)$. This is 
 illustrated in Fig.~\ref{fig:diversity9} by simulating the performance of the LC-BRD and the approximated centralized  solution as a function of the noise variance spread in the network over all targets, i.e.,  $\textrm{max}_{i,j}( \sigma^{(i)}_{r_j}) / \textrm{min}_{i,j}( \sigma^{(i)}_{r_j}) $. 

So far, the focus has been on the scenarios where all radars had the same interests, full observability and the radar network was fully connected. Now, let us remove these restrictions  by firstly analyzing the scenario where the observability of each radar varies between $3$ and $5$ and the radar connectivity for some radars is not full while all radars are still interested in all targets. We set $T=5$ targets, $N=6$ radars, $m=1$  measurements per scan  and update time $t_u=0.25\,s$. Note that the compared algorithms were modified accordingly in order to take into account the above constraints. As it can be seen in Fig.~\ref{fig:1_radar_1_target}, the LC-BRD algorithm clearly outperforms the distributed random one ($K=1$) due to the fact that the equally balanced target allocations (from a single radar perspective) are not necessary reasonable, in contrast to the scenario in Figs.~\ref{fig:7}-\ref{fig:mogucnost_8}.

			\begin{figure}
				\centering 

				\subfigure[Interest: All targets]{
					\includegraphics[width=0.34\textwidth]{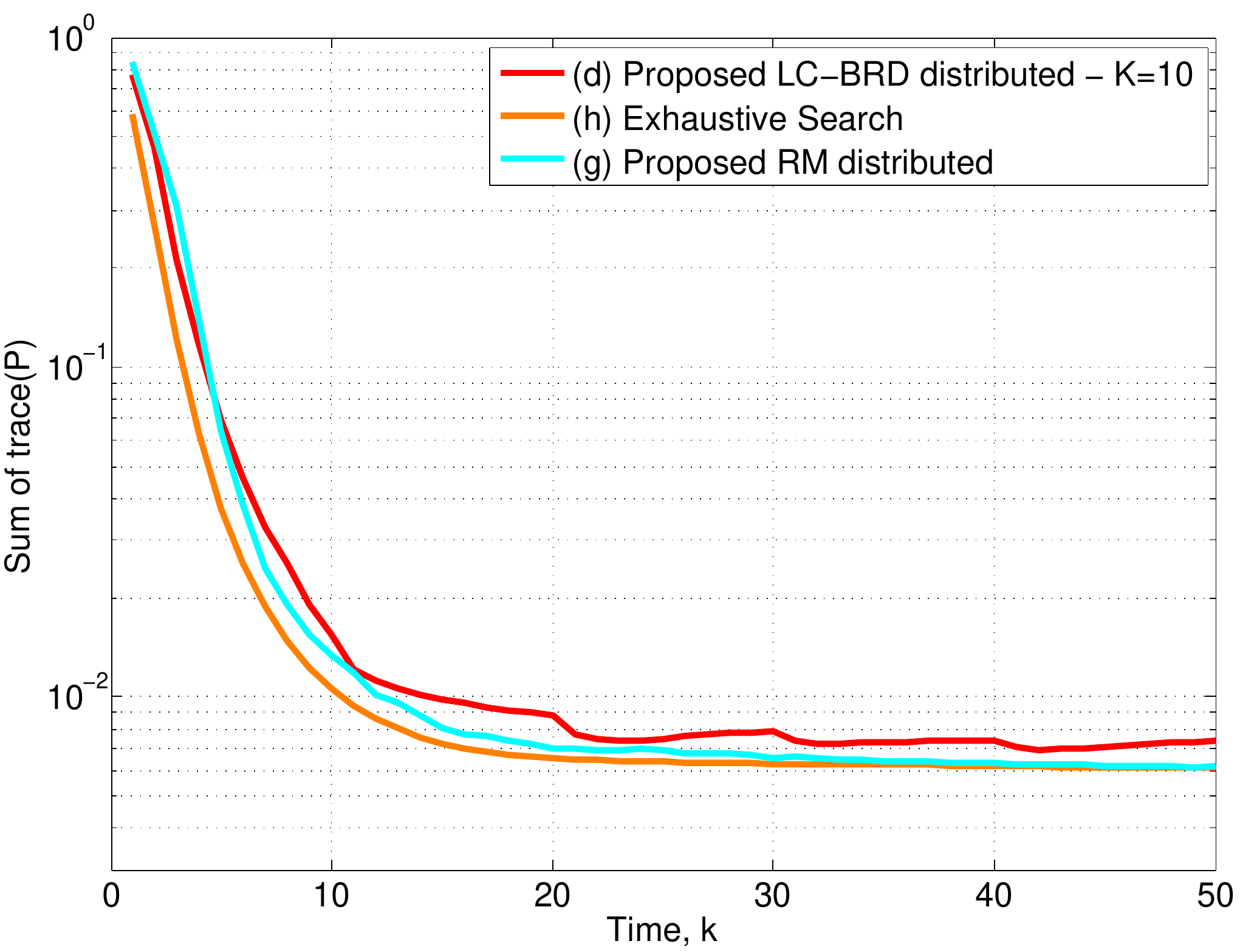}
					\label{fig:Exhaustive_RM_BRDa}
				} \hspace{0.05cm}
				\subfigure[Interest: 4 targets]{
					\includegraphics[width=0.34\textwidth]{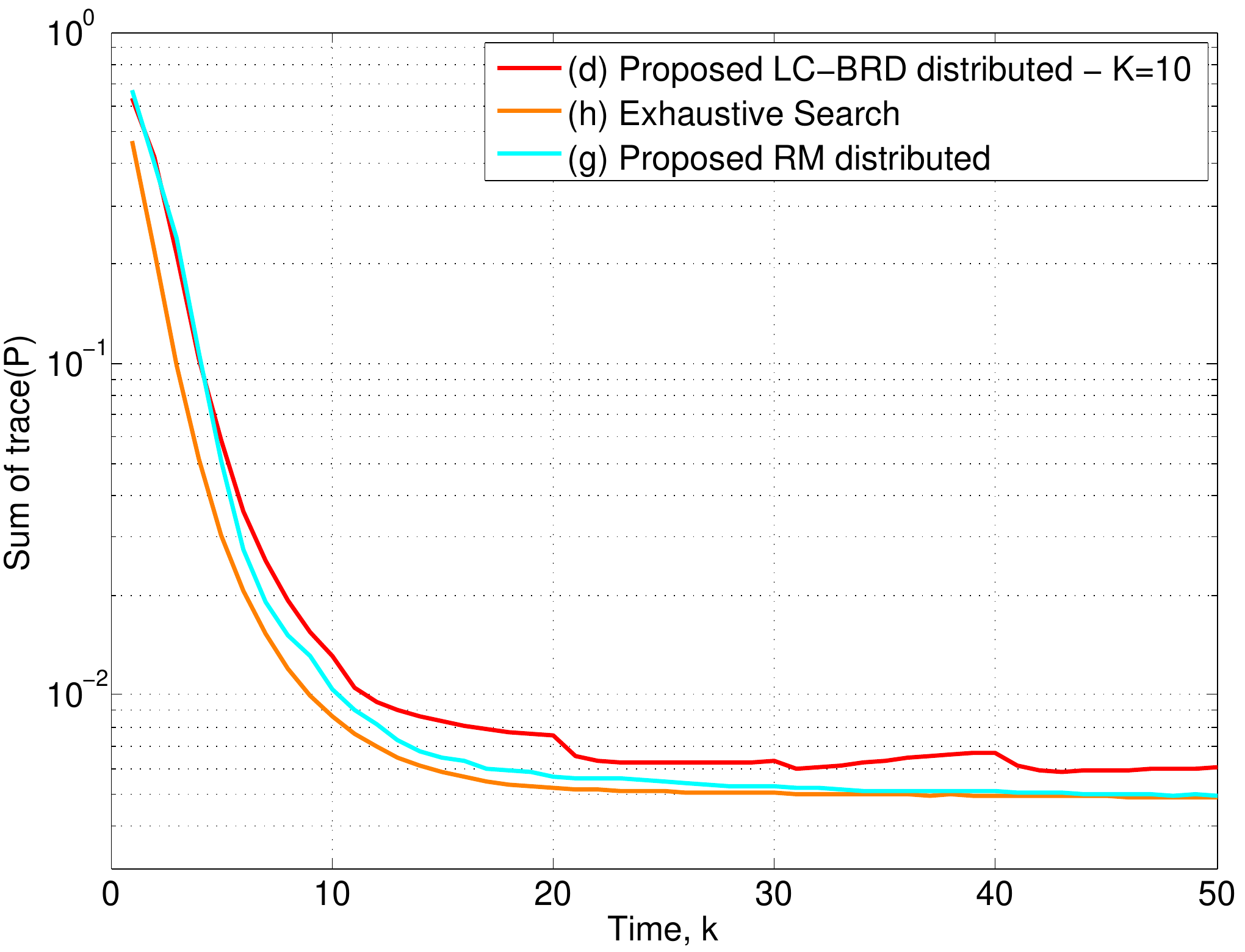}
																
					\label{fig:Exhaustive_RM_BRDb}
				}
				\caption{Scenario with  $  \vert\mathcal{T}_i\vert= 4 $ and $ \mathcal{N}_i= \mathcal{N}$, with $T=5$, $N=6$, $m=1$  and $t_u=0.25\,s$. 
					The case where all radars have the same interests (all targets) is in~\subref{fig:Exhaustive_RM_BRDa}, and the case where their interests differ in general is  in~\subref{fig:Exhaustive_RM_BRDb}.}
				\label{fig:Exhaustive_RM_BRD}
			\end{figure}

Finally, we compare the proposed strategies with the centralized solution based on exhaustive search, i.e.,
\begin{itemize}
	\item[h)] \textit{Exhaustive Search} - The centralized search is implemented with full knowledge of all radars' interests, observability and connectivity conditions and at each time instant the best allocation optimizing the sum of all radars' utilities is selected. 
\end{itemize}
Figure~\ref{fig:Exhaustive_RM_BRD} shows that the LC-BRD performs well given its complexity. Note also that for the case where the interests of the radars are not necessary the same four targets, there are no theoretical guarantees that the NE exist(s) nor that a BRD-based algorithm may achieve an NE point; however, the LC-BRD still preforms relatively well. On the other hand, the RM-based algorithm, which is designed for more general scenarios, performs better than the LC-BRD and it closely approaches the centralized, exhaustive search solution.

\section{Conclusions}
\label{sec:conc}


In this article, we have proposed a new formulation of the track selection problem for a multi-target tracking scenario in an MFR network using the non-cooperative games. The target selections of each radar are considered to be autonomous; there is no central entity to tell radars what to do nor is there any negotiation process among the radars. We have analyzed two indicative scenarios with equal and heterogeneous conditions of observability and connectivity as well as radar interests. In the former scenario, the  Nash equilibria of the underlying anti-coordination games have been analyzed and 
a simple yet effective distributed algorithm %
that introduces a balancing effect in track selections 
has been proposed. Afterwards, for a more demanding scenario, the solution concept of correlated equilibria has been employed and a more sophisticated, distributed algorithm based on the regret-matching has been  proposed. 
Finally, computer simulations have verified that both proposed algorithms closely approximate the centralized solution while mitigating its inherent complexity.

\bibliographystyle{IEEEtran}
\bibliography{IEEEabrv,./Nikola2_radar_journ}

\end{document}